\documentclass[aps,notitlepage,twocolumn,nofootinbib,superscriptaddress,longbibliography]{revtex4-1}
\usepackage{graphicx}
\usepackage{amsmath,amsfonts,amssymb}
\usepackage{mathtools}
\usepackage{multirow}
\usepackage[utf8]{inputenc}
\usepackage[T1]{fontenc}

\usepackage{bm}
\usepackage{algorithm,algorithmic}
\usepackage{mathtools}
\usepackage{mathrsfs}
\usepackage[shortlabels]{enumitem}
\usepackage{graphicx,epic,eepic,epsfig,amsmath,latexsym,amssymb,verbatim,color}
\usepackage{amsfonts}  % blackboard math symbols
\usepackage{bbm}
 
\usepackage{float}
\usepackage{tikz}
\usetikzlibrary{chains}
\usetikzlibrary{fit}
\usepackage{pgflibraryarrows}		%optional
\usepackage{pgflibrarysnakes}		%optional

\usepackage{epsfig}
\usetikzlibrary{shapes.symbols,patterns} % for source symbols
\usepackage{pgfplots}

\usepackage[strict]{changepage}
\usepackage{hyperref}
\hypersetup{colorlinks=true,citecolor=blue,linkcolor=blue,filecolor=blue,urlcolor=blue,breaklinks=true}

\usepackage[marginal]{footmisc}
\usepackage{url}
\usepackage{theorem}
\newtheorem{definition}{Definition}
\newtheorem{proposition}{Proposition}
\newtheorem{lemma}[proposition]{Lemma}

\newtheorem{theorem}[proposition]{Theorem}
\newtheorem{corollary}[proposition]{Corollary}

\def\squareforqed{\hbox{\rlap{$\sqcap$}$\sqcup$}}
\def\qed{\ifmmode\squareforqed\else{\unskip\nobreak\hfil
\penalty50\hskip1em\null\nobreak\hfil\squareforqed
\parfillskip=0pt\finalhyphendemerits=0\endgraf}\fi}
\def\endenv{\ifmmode\;\else{\unskip\nobreak\hfil
\penalty50\hskip1em\null\nobreak\hfil\;
\parfillskip=0pt\finalhyphendemerits=0\endgraf}\fi}
\newenvironment{proof}{\noindent \textbf{{Proof~} }}{\hfill $\blacksquare$}

\newcounter{remark}

\newcounter{example}

\mathchardef\ordinarycolon\mathcode`\:
\mathcode`\:=\string"8000
\def\vcentcolon{\mathrel{\mathop\ordinarycolon}}
\begingroup \catcode`\:=\active
  \lowercase{\endgroup
  \let :\vcentcolon
  }

\usepackage{cleveref}
\usepackage{graphicx}
\usepackage{xcolor}

\definecolor{darkblue}{RGB}{0,76,156}
\definecolor{darkkblue}{RGB}{0,0,153}
\definecolor{blue2}{RGB}{102,178,255}
\definecolor{darkred}{RGB}{195,0,0}

\RequirePackage[framemethod=default]{mdframed}
\newmdenv[skipabove=7pt,
skipbelow=7pt,
backgroundcolor=darkblue!15,
innerleftmargin=5pt,
innerrightmargin=5pt,
innertopmargin=5pt,
leftmargin=0cm,
rightmargin=0cm,
innerbottommargin=5pt,
linewidth=1pt]{tBox}

\newmdenv[skipabove=7pt,
skipbelow=7pt,
backgroundcolor=blue2!25,
innerleftmargin=5pt,
innerrightmargin=5pt,
innertopmargin=5pt,
leftmargin=0cm,
rightmargin=0cm,
innerbottommargin=5pt,
linewidth=1pt]{dBox}

\newmdenv[skipabove=7pt,
skipbelow=7pt,
backgroundcolor=darkred!15,
innerleftmargin=5pt,
innerrightmargin=5pt,
innertopmargin=5pt,
leftmargin=0cm,
rightmargin=0cm,
innerbottommargin=5pt,
linewidth=1pt]{rBox}

\newcommand{\nc}{\newcommand}
\nc{\bra}[1]{\langle#1|}
\nc{\ket}[1]{|#1\rangle}
\nc{\ketbra}[2]{|#1\rangle\!\langle#2|}
\nc{\braket}[2]{\langle#1|#2\rangle}
\newcommand{\braandket}[3]{\langle #1|#2|#3\rangle}
\nc{\proj}[1]{| #1\rangle\!\langle #1 |}
\nc{\avg}[1]{\langle#1\rangle}
\nc{\rank}{\operatorname{Rank}}
\nc{\smfrac}[2]{\mbox{$\frac{#1}{#2}$}}
\nc{\tr}{\operatorname{tr}}
\nc{\ox}{\otimes}
\nc{\dg}{\dagger}
\nc{\dn}{\downarrow}
\nc{\var}{{\operatorname{var}}}
\nc{\rar}{\rightarrow}
\nc{\lrar}{\longrightarrow}
\nc{\argmin}{{\operatorname{argmin}}}
\nc{\id}{{\operatorname{id}}}

\nc{\Choi}{Choi-Jamio\l{}kowski}

\usepackage{hyperref}
\hypersetup{colorlinks=true,citecolor=blue,linkcolor=blue,filecolor=blue,urlcolor=blue,breaklinks=true}

\makeatletter
\def\grd@save@target#1{%
  \def\grd@target{#1}}
\def\grd@save@start#1{%
  \def\grd@start{#1}}
\tikzset{
  grid with coordinates/.style={
    to path={%
      \pgfextra{%
        \edef\grd@@target{(\tikztotarget)}%
        \tikz@scan@one@point\grd@save@target\grd@@target\relax
        \edef\grd@@start{(\tikztostart)}%
        \tikz@scan@one@point\grd@save@start\grd@@start\relax
        \draw[minor help lines,magenta] (\tikztostart) grid (\tikztotarget);
        \draw[major help lines] (\tikztostart) grid (\tikztotarget);
        \grd@start
        \pgfmathsetmacro{\grd@xa}{\the\pgf@x/1cm}
        \pgfmathsetmacro{\grd@ya}{\the\pgf@y/1cm}
        \grd@target
        \pgfmathsetmacro{\grd@xb}{\the\pgf@x/1cm}
        \pgfmathsetmacro{\grd@yb}{\the\pgf@y/1cm}
        \pgfmathsetmacro{\grd@xc}{\grd@xa + \pgfkeysvalueof{/tikz/grid with coordinates/major step}}
        \pgfmathsetmacro{\grd@yc}{\grd@ya + \pgfkeysvalueof{/tikz/grid with coordinates/major step}}
        \foreach \x in {\grd@xa,\grd@xc,...,\grd@xb}
        \node[anchor=north] at (\x,\grd@ya) {\pgfmathprintnumber{\x}};
        \foreach \y in {\grd@ya,\grd@yc,...,\grd@yb}
        \node[anchor=east] at (\grd@xa,\y) {\pgfmathprintnumber{\y}};
      }
    }
  },
  minor help lines/.style={
    help lines,
    step=\pgfkeysvalueof{/tikz/grid with coordinates/minor step}
  },
  major help lines/.style={
    help lines,
    line width=\pgfkeysvalueof{/tikz/grid with coordinates/major line width},
    step=\pgfkeysvalueof{/tikz/grid with coordinates/major step}
  },
  grid with coordinates/.cd,
  minor step/.initial=.2,
  major step/.initial=1,
  major line width/.initial=2pt,
}
\makeatother

\usepackage{thmtools}
\usepackage{thm-restate}
\usepackage{etoolbox}
\makeatletter
\def\problem@s{}
\newcounter{problems@cnt}

\newcommand{\allproblems}{\problem@s}
\makeatother

\usepackage{tcolorbox}
\usepackage{relsize}
\usepackage{booktabs}

\usepackage[qm]{qcircuit}
\usepackage{array}

\usepackage{tikz}
\usetikzlibrary{positioning}
\usetikzlibrary{shapes.geometric}
\usetikzlibrary{calc}
\definecolor{tensorblue}{rgb}{0.8,0.9,1}
\tikzset{ten/.style={fill=tensorblue}}
\newcommand{\diagram}[1]{ \begin{array}{cc}\begin{tikzpicture}[scale=.5,every node/.style={sloped,allow upside down},baseline={([yshift=+0ex]current bounding box.center)}] #1 \end{tikzpicture} \end{array} }
\usepackage{extarrows}

\allowdisplaybreaks
\begin{document}

\title{Fundamental limitations on optimization in variational quantum algorithms}
\author{Hao-Kai Zhang}
\affiliation{Institute for Quantum Computing, Baidu Research, Beijing 100193, China}
\affiliation{Institute for Advanced Study, Tsinghua University, Beijing 100084, China}
\author{Chengkai Zhu}
\affiliation{Institute for Quantum Computing, Baidu Research, Beijing 100193, China}
\author{Geng Liu}
\affiliation{Institute for Quantum Computing, Baidu Research, Beijing 100193, China}
\author{Xin Wang}
\email{wangxin73@baidu.com}
\affiliation{Institute for Quantum Computing, Baidu Research, Beijing 100193, China}

\begin{abstract}
Exploring quantum applications of near-term quantum devices is a rapidly growing field of quantum information science with both theoretical and practical interests. A leading paradigm to establish such near-term quantum applications is variational quantum algorithms (VQAs). These algorithms use a classical optimizer to train a parameterized quantum circuit to accomplish certain tasks, where the circuits are usually randomly initialized. In this work, we prove that for a broad class of such random circuits, the variation range of the cost function via adjusting any local quantum gate within the circuit vanishes exponentially in the number of qubits with a high probability. This result can unify the restrictions on gradient-based and gradient-free optimizations in a natural manner and reveal extra harsh constraints on the training landscapes of VQAs. Hence a fundamental limitation on the trainability of VQAs is unraveled, indicating the essential mechanism of the optimization hardness in the Hilbert space with exponential dimension. We further showcase the validity of our results with numerical simulations of representative VQAs. We believe that these results would deepen our understanding of the scalability of VQAs and shed light on the search for near-term quantum applications with advantages.
\end{abstract}

\date{\today}
\maketitle

%\tableofcontents

%%%%%%%%%%%%%%%%%%%%%%%%%%%%%%%%%%%%%%%%%%%%%%%%%%%%%%%%%%%%%%%%%%%%
Enormous efforts have been made to develop noisy intermediate scale quantum (NISQ) devices~\cite{Preskill2018} toward achieving near-term quantum advantage for practical applications in key areas including many-body physics~\cite{Wecker2015,Ho2018,Uvarov2020}, chemistry~\cite{McArdle2018a}, finance~\cite{Egger2020,Herman2022,Bouland2020}, and machine learning~\cite{Biamonte2017b}. The hybrid quantum-classical computation framework, including variational quantum algorithms (VQAs)~\cite{McClean2016,Cerezo2021a,Bharti2021,Endo2020}, is widely believed to be promising in making use of NISQ devices to deliver meaningful quantum applications.  Specifically, VQAs use a classical optimizer to train a parameterized quantum circuit (PQC) in order to solve problems in various topics such as ground state preparation~\cite{Peruzzo2014}, quantum linear algebra~\cite{Xu2019a,Huang2019b,Bravo-Prieto2019,Wang2020d}, quantum metrology~\cite{Beckey2022,Koczor2020,Meyer2021}, quantum entanglement~\cite{Wang2020,Bravo-Prieto2019a,Chen2021,Zhao2021}, and machine learning~\cite{Schuld2018a,LaRose2020,Schuld2021}. 

With the aim to outperform classical algorithms and show quantum advantage on certain tasks, a critical issue is whether VQAs can be extended to solve large-scale systems, i.e., the scalability of VQAs. Unfortunately, many studies point out that training in VQAs requires exponential resources with the system size under certain conditions~\cite{McClean2018,Arrasmith2020,Cerezo2021c,Arrasmith2021,Wang2021,Holmes2021,Bittel2021,OrtizMarrero2021,StilckFranca2021,Uvarov2021,Campos2021,DePalma2022}. Besides the practical limitations such as noises~\cite{Wang2021}, even ideal quantum devices will suffer from the so-called \textit{barren plateau} phenomenon~\cite{McClean2018}. It was shown that the gradient of the cost function vanishes exponentially in the number of qubits with a high probability for a random initialized PQC with sufficient depth, analogous to the vanishing gradient issue in classical neural networks. Consequently, exponentially vanishing gradients demand an exponential precision in the cost function measurement on a quantum device~\cite{Knill2006} to make progress in the gradient-based optimization, and hence an exponential complexity in the number of qubits.

\begin{figure}
    \includegraphics[width=\columnwidth]{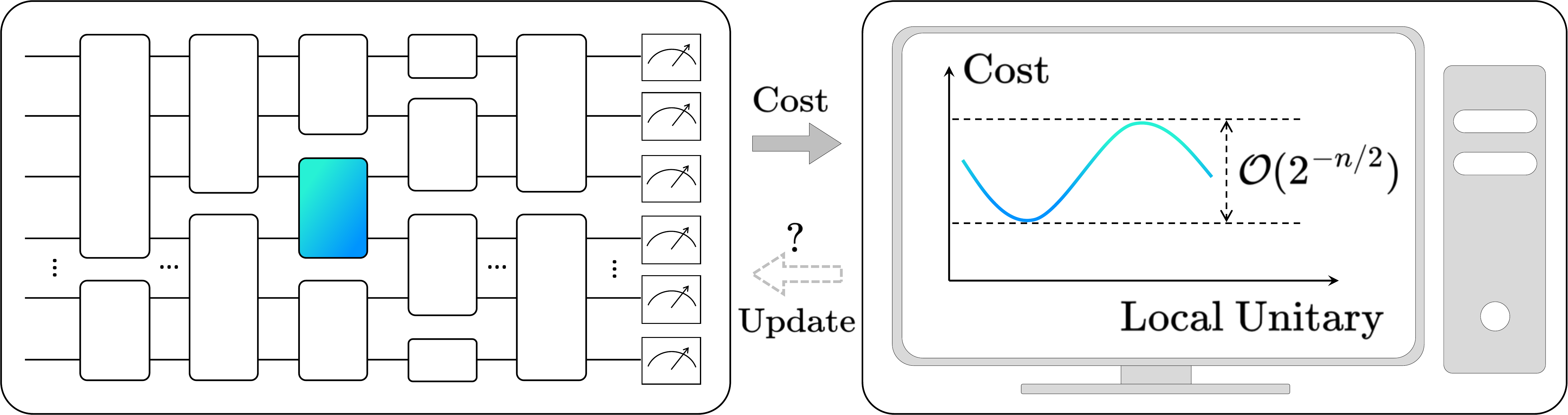}
    \caption{
        \textbf{Summary of the main results.} The left part depicts a randomly initialized PQC on $n$ qubits forming a $2$-design. The right part symbolically depicts the cost function on a classical device vs. the local unitary highlighted in the left part. The two parts together show a generic VQA routine: training a quantum circuit using a classical optimizer. This work proves that the cost function will fluctuate in an exponentially small range in the number of qubits with a high probability when we vary an arbitrary local unitary within the circuit.
    }
    \label{fig:luo_main}
\end{figure}

Several attempts have been made to avoid barren plateaus, such as higher order derivatives~\cite{Huembeli2021}, gradient-free optimizers including gate-by-gate optimization~\cite{Nakanishi2019,Ostaszewski2019}, proper initialization~\cite{Grant2019}, pre-training including adaptive methods~\cite{Verdon2019,Grimsley2019,Zhang2021,Skolik2021,Grimsley2022}, circuit architectures~\cite{Pesah2021,Liu2022a} and cost function choices~\cite{Cerezo2021,Kieferova2021}. More efforts are needed to study the general effectiveness of these attempts~\cite{Arrasmith2020,Cerezo2021c} and develop new strategies to improve the trainability and scalability of VQAs. As a guidance for exploring effective training strategies, it is crucial to uncover the essential mechanisms behind the barren plateau phenomenon.

However, few rigorous scaling results are known for generic VQAs besides gradient analyses and their descendent~\cite{Arrasmith2020,Cerezo2021c,Arrasmith2021}. It would be quite helpful for designing efficient algorithms if we could gain information on the training landscape beyond gradients. Naturally, we would like to know the entire variation range of the cost function when adjusting a single~\cite{Nakanishi2019,Ostaszewski2019} or several parameters as a guidance for the optimization, instead of just the limited information of the vicinity from gradient analyses. Combined with the fact that parameters usually enter the circuit independently through local quantum gates, e.g., the single-qubit rotation gates, all of which motivate our work where we are chiefly concerned with the variation range of the cost function via varying a local unitary within a quantum circuit.

In this work, we present a new rigorous scaling theorem on the trainability of VQAs beyond gradients. As summarized in Fig.~\ref{fig:luo_main}, we prove that when varying a local unitary within a sufficiently random circuit, the expectation and variance of the variation range of the cost function vanish exponentially in the number of qubits. Then through simple derivations, we show that this theorem implies exponentially vanishing gradients and cost function differences, and hence unifies the restrictions on gradient-based and gradient-free optimizations. Meanwhile, this theorem further delivers extra meaningful information about the training landscapes and optimization possibilities of VQAs. In this sense, we obtain a fundamental limitation on optimization in VQAs. Next we illustrate the applications of our theorem on representative VQAs. A tighter bound for the fidelity-type cost function is provided specifically even for shallow random circuits. At last, we perform numerical simulations on these representative VQAs, where the scaling exponents coincide with our analytical results almost precisely.

\vspace{3mm}

\noindent\textbf{Results}\\
\textbf{Limitations of local unitary optimization.} 
We start by introducing a general setting of VQAs used throughout our analysis. VQAs usually use a classical optimizer to train a quantum circuit $\mathbf{U}$ with an input state $\rho$ by minimizing a task-dependent cost function $C$, which is typically chosen as the expectation value of some Hermitian operator $H$
\begin{equation}\label{Eq:VQA_cost_def}
    C_{H,\rho}(\mathbf{U}) = \tr(H \mathbf{U}\rho \mathbf{U}^\dagger).
\end{equation}
Divide the whole qubit system into two parts $A, B$ with $m$ qubits and $n-m$ qubits, respectively. Here $m$ is a fixed constant not scaling with $n$ so that we call $A$ a local subsystem. The circuit $\mathbf{U}$ is often composed of local unitaries on real devices, such as the single-qubit rotation gates and the CNOT gate. We focus on a local unitary $U_A$ within $\mathbf{U}$ acting on subsystem $A$. As shown in Fig.~\ref{fig:main-circuit}, we denote the sub-circuit of $\mathbf{U}$ before $U_A$ as $V_1$ and that behind $U_A$ as $V_2$, such that $\mathbf{U}=V_2(U_A\otimes I_B)V_1$ where $I_B$ is the identity operator on $B$. $V_1$, $V_2$ and $U_A$ are independent of each other. 
% Unless otherwise specified, we will use these notations throughout this paper.
% Once the value of $C_{H,\rho}(\mathbf{U})$ is evaluated on quantum devices, classical optimizers are used to minimize it using various strategies. There are gradient-based or gradient-free optimization methods whose spirit is to make the cost function approach the minimum as possible in each step of the optimization procedure when the circuit is parameterized with some parameter $\theta$. And what we commonly do in a step is to adjust local quantum gates by the updated $\theta$, which necessarily determines how far our optimization process can go. 
% To characterize how much difference in the cost function one can make by adjusting local quantum gates, we introduce the \textit{variation range of the cost function} when varying a local unitary.

To characterize the training landscape beyond the limited information of the vicinity from gradient analyses, we introduce a central quantity throughout this work, i.e., the \textit{variation range of the cost function} via varying a local unitary.
\begin{definition}
For a generic VQA cost function $C_{H,\rho}(\mathbf{U})$ in Eq.~{\rm (\ref{Eq:VQA_cost_def})}, we define its variation range with given $V_1,V_2$ as
\begin{equation}\label{Eq:def_delta_V1V2}
    \Delta_{H,\rho}(V_1,V_2) := \max_{U_A} C_{H,\rho}(\mathbf{U}) - \min_{U_A} C_{H,\rho}(\mathbf{U}),
\end{equation}
where the maximum and minimum with respect to $U_A$ are taken over the unitary group $\mathcal{U}(2^m)$ of degree $2^m$.
\end{definition}

The quantity $\Delta_{H,\rho}(V_1,V_2)$ intuitively reflects the maximal possible influence that the local unitary $U_A$ can have on the VQA cost function. We establish an upper bound on $\Delta_{H,\rho}(V_1,V_2)$ in the sense of probability by Theorem~\ref{theorem:main-theorem}, which thus delivers a limitation on optimizing an arbitrary local unitary. To be specific, we prove that if either $V_1$, $V_2$, or both match the Haar distribution up to the second moment, i.e., are sampled from unitary 2-designs~\cite{Dankert2009}, the expectation of $\Delta_{H,\rho}(V_1,V_2)$ vanishes exponentially in the number of qubits. Supplementary Note \textcolor{blue}{1} introduce some preliminaries on unitary designs. The proofs of our results are sketched in the Methods and detailed in the Supplementary Information.

\begin{figure}
    \centering
    \includegraphics[width=\columnwidth]{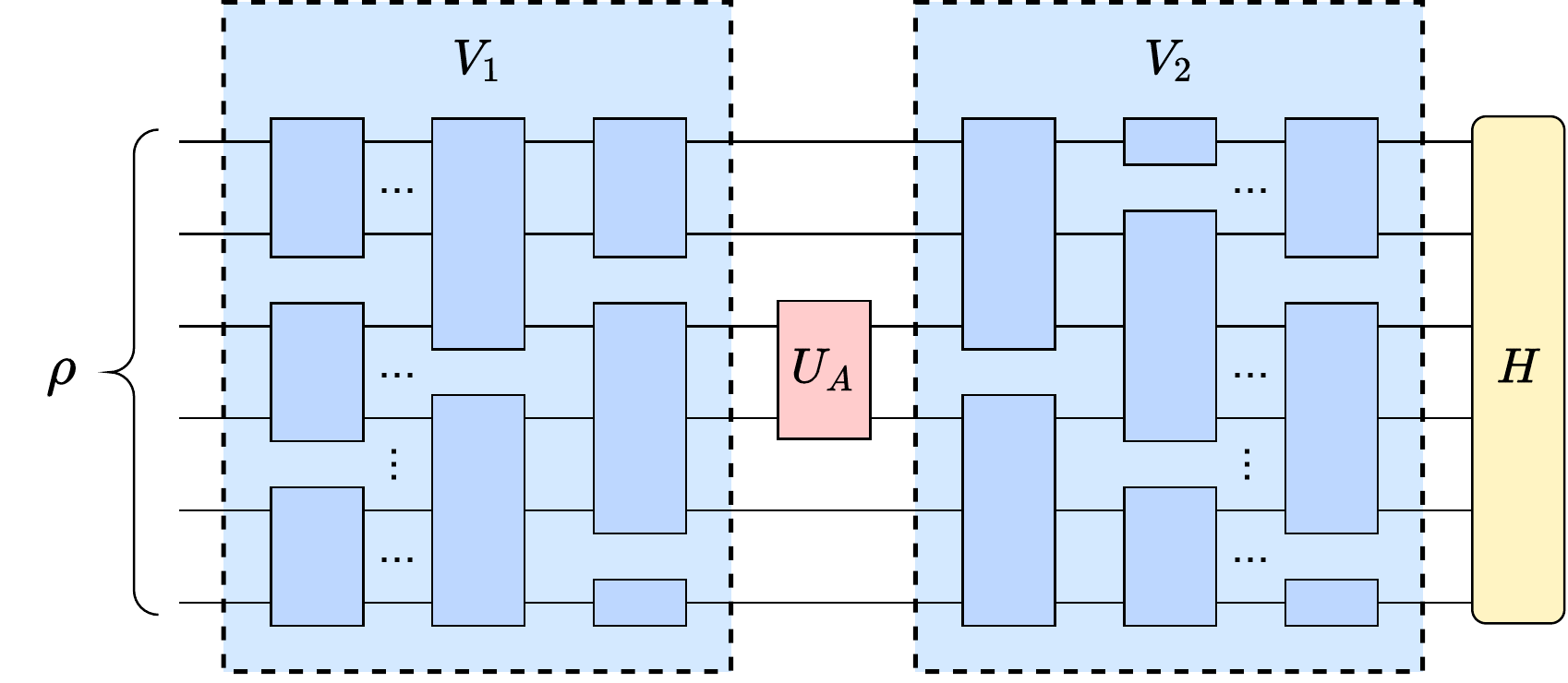}
    \caption{\textbf{Partition of the quantum circuit in our analysis}. A parameterized quantum circuit used in a VQA with an $n$-qubit input quantum state $\rho$ (either pure or mixed). We denote $U_A$ as a tunable local unitary implemented by some local quantum gates. Then the left part of the circuit implements unitary $V_1$, and the right part implements unitary $V_2$. A generic cost function of a VQA is the expectation value over some objective operator $H$.}
    \label{fig:main-circuit}
\end{figure}

\begin{theorem}\label{theorem:main-theorem}
Suppose $\mathbb{V}_1, \mathbb{V}_2$ are ensembles from which $V_1, V_2$ are sampled, respectively. If either $\mathbb{V}_1$ or $\mathbb{V}_2$, or both form unitary $2$-designs, then for arbitrary $H$ and $\rho$, the following inequality holds
\begin{equation}\label{Eq:main-theorem}
    \mathbb{E}_{V_1, V_2} [ \Delta_{H,\rho}(V_1,V_2) ] \leq \frac{ w(H) }{2^{n/2-3m-2}},
\end{equation}
where $\mathbb{E}_{V_1,V_2}$ denotes the expectation over $\mathbb{V}_1,\mathbb{V}_2$ independently. $w(H) = \lambda_{\max}(H) - \lambda_{\min}(H)$ denotes the spectral width of $H$, where $\lambda_{\max}(H)$ is the maximum eigenvalue of $H$ and $\lambda_{\min}(H)$ is the minimum. 
\end{theorem}

We make several remarks on Theorem~\ref{theorem:main-theorem}. Firstly, due to the non-negativity and boundedness of the variation range, i.e., $\Delta_{H,\rho}\in[0, w(H)]$, the variance of $\Delta_{H,\rho}$ can be bounded by its expectation times $w(H)$. Thus from Theorem~\ref{theorem:main-theorem} we know that the variance also vanishes exponentially
\begin{equation}\label{Eq:VarV1V2_delta}
    \operatorname{Var}_{V_1,V_2}[\Delta_{H,\rho}(V_1,V_2)] 
    \leq \frac{ w^2(H) }{2^{n/2-3m-2}}.
\end{equation}
Note that $w(H)\in\mathcal{O}(poly(n))$ holds for common VQAs. Moreover, Theorem~\ref{theorem:main-theorem} together with Markov's inequality provides an upper bound of the probability that $\Delta_{H,\rho}(V_1,V_2)$ deviates from zero. Namely, the following concentration inequality
\begin{equation}\label{Eq:PrV1V2_delta}
    \operatorname{Pr}[\Delta_{H,\rho}(V_1,V_2) \geq \epsilon] \leq  \frac{1}{\epsilon} \cdot \frac{w(H)}{2^{n/2-3m-2}},
\end{equation}
holds for any $\epsilon>0$, which means that the probability that the variation range $\Delta_{H,\rho}$ is non-zero to some fixed precision is exponentially small in the number of qubits.

Secondly, the exponentially small bound in (\ref{Eq:main-theorem}) is still non-trivial when $U_A$ is a \textit{global unitary} and satisfies the parameter-shift rule~\cite{Guerreschi2017,Mitarai2018,Schuld2018,Crooks2019,Mari2021} if both $\mathbb{V}_1$ and $\mathbb{V}_2$ form $2$-designs. Suppose $U_A = e^{-i\theta\Omega}$ with the Hermitian generator $\Omega$ satisfying $\Omega^2=I$. Since $\Omega$ has only two different eigenvalues $\pm 1$, there exists a unitary $W$ such that $We^{-i\theta\Omega}W^\dagger$ becomes a local unitary acting on a single qubit non-trivially. $W$ and $W^\dagger$ could be absorbed into $2$-design ensembles with $W^\dagger \mathbb{V}_1$ and $\mathbb{V}_2 W$ still forming $2$-designs~\cite{Kaznatcheev2009}. Therefore, the proof for global unitaries satisfying the parameter-shift rule can be reduced back to the case of local unitaries.

Moreover, it is worth noticing that the compact bound in (\ref{Eq:main-theorem}) only involves the spectral width $w(H)$ and does not depend on any detail of the Hermitian operator $H$. But if some specific structures about $H$ are known, e.g., the Pauli decomposition of $H$, a tighter bound could be derived in Supplementary Note~\textcolor{blue}{2} which depends on the coupling complexity of $H$. In addition, if the cost function reduces to the form of the fidelity between pure states, we could have a tighter bound with scaling $\mathcal{O}(2^{-n})$ in Proposition~\ref{prop:state-learning} below. Theorem~\ref{theorem:main-theorem} can be generalized to arbitrary dimensions besides qubit systems of dimension $2^n$, e.g., qutrit and qudit systems. A detailed proof is provided in Supplementary Note~\textcolor{blue}{2}. Finally, we point out that local operations making small influences on the whole system is a physically natural but mathematically non-trivial argument. For instance, it is easy to prove that even a single-qubit unitary is enough to rotate an arbitrary $n$-qubit pure state to a new state with a zero fidelity with the original one, which is a practical example that local operations make a great influence. So Theorem~\ref{theorem:main-theorem} may be invaluable as a rigorous version of the above argument in the context of VQAs and random quantum circuits.

\vspace{2mm}
\noindent\textbf{Implications of Theorem~{\rm\ref{theorem:main-theorem}}.}
Here we briefly demonstrate how Theorem~\ref{theorem:main-theorem} implies the restrictions on both gradient-based~\cite{McClean2018,Cerezo2021c} and gradient-free optimizations~\cite{Arrasmith2020} in a more natural manner, and indicates the extra restrictions besides them. In the following we focus on a PQC applicable for Theorem~\ref{theorem:main-theorem} with $M$ trainable parameters $\{\theta_{\mu}\}_{\mu=1}^M$ and denote the variation range of the cost function via varying $\theta_\mu$ as $\Delta_\mu$.

Consider the gradient-based optimization first. On the one hand, in the case where the parameter-shift rule is valid~\cite{Guerreschi2017,Mitarai2018,Schuld2018,Crooks2019,Mari2021}, Theorem~\ref{theorem:main-theorem} can strictly deduce vanishing gradients. Suppose $\{\theta_{\mu}\}_{\mu=1}^M$ are applicable for the parameter-shift rule (e.g., hardware-efficient ansatzes). Namely, $\theta_\mu$ enters the unitary $e^{-i\theta_\mu\Omega_\mu}$ within the circuit where $\Omega_\mu$ is a Hermitian generator satisfying $\Omega_\mu^2=I$. From Theorem~$\ref{theorem:main-theorem}$ we know that the expectation of $\Delta_{\mu}$ vanishes exponentially. Therefore, the derivative $\partial_\mu C := \frac{\partial C}{\partial \theta_\mu}$ with respect to $\theta_{\mu}$ satisfies
\begin{equation}
    \begin{aligned}
        \mathbb{E}[\left| \partial_\mu C \right|] 
        & = \mathbb{E}\left[ \left| C\left(\bm{\theta} + \frac{\pi}{4} \bm{{\rm e}}_\mu \right) - C\left(\bm{\theta} - \frac{\pi}{4} \bm{{\rm e}}_\mu \right) \right| \right] \\ 
        & \leq \mathbb{E}[\Delta_{\mu}] \in \mathcal{O}(2^{-n/2}),
    \end{aligned}
\end{equation}
where $\bm{{\rm e}}_{\mu}$ is the unit vector in the parameter space corresponding to $\theta_{\mu}$. From Markov's inequality as in (\ref{Eq:PrV1V2_delta}), we know that the probability that the derivative $\partial_\mu C$ deviates from zero by a small constant is exponentially small in the number of qubits. 

\begin{figure}
    \includegraphics[width=\columnwidth]{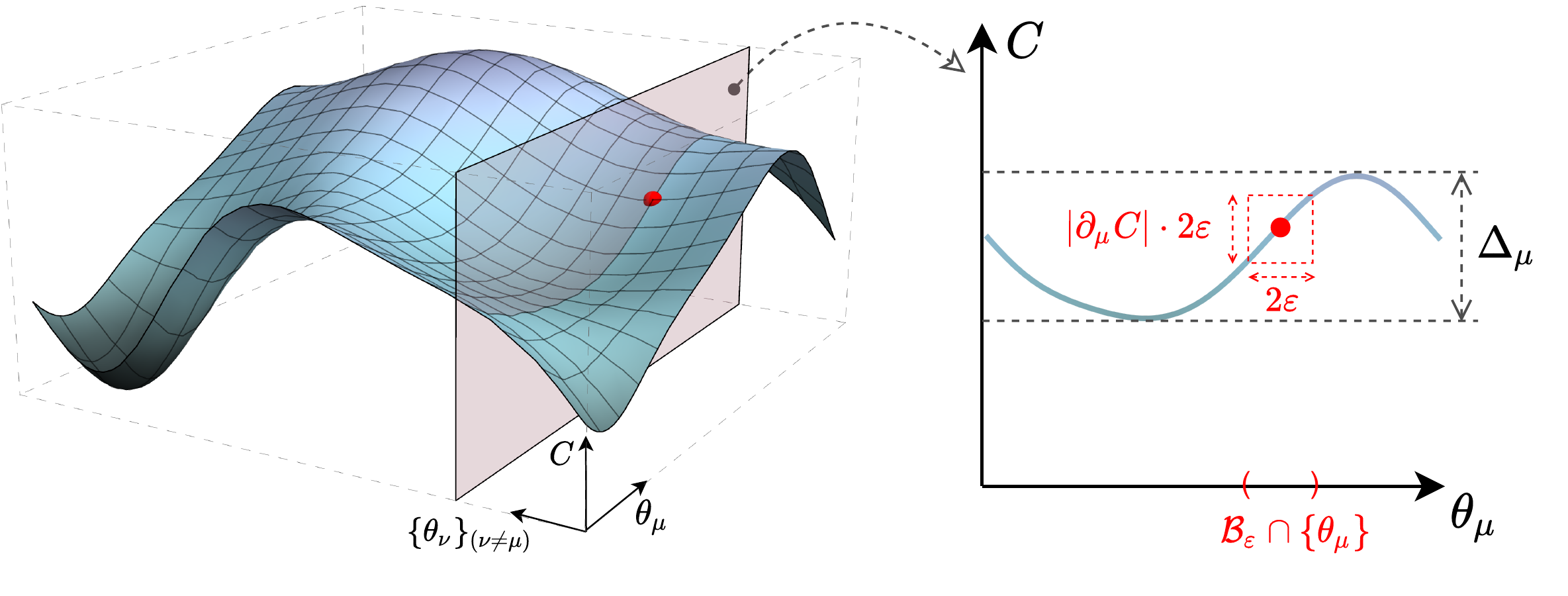}
    \caption{
    \textbf{Sketch of our results implying vanishing gradients.} The left panel sketches the whole training landscape with one of the parameters $\theta_\mu$ as the $x$-axis, all of other parameters $\{\theta_\nu\}_{\nu\neq\mu}$ as the $y$-axis symbolically and the cost function value $C$ as the $z$-axis. The right panel depicts a typical sample of the $z$-$x$ cross-section from the landscape on the left with variation range $\Delta_\mu$. Up to the linear approximation error, $\Delta_\mu$ can serve as an upper bound for the absolute derivative $|\partial_\mu C|$ times the vicinity size $2\varepsilon$. Thus Theorem~\ref{theorem:main-theorem} implies vanishing gradients even in the absence of the parameter-shift rule.
    }
    \label{fig:relation-BP}
\end{figure}

On the other hand, even in the absence of the parameter-shift rule, vanishing gradients could still be obtained approximately by the following arguments. Consider the vicinity of a random initialized parameter point where the linear approximation error is negligible, denoted as an $\varepsilon$-ball $\mathcal{B}_{\varepsilon}$ of radius $\varepsilon$ (here $\varepsilon$ plays the same role as the learning rate). As shown in Fig. \ref{fig:relation-BP}, the linearity in $\mathcal{B}_{\varepsilon}$ together with Theorem~\ref{theorem:main-theorem} leads to
\begin{equation}
    \mathbb{E}\left[|\partial_{\mu}C|\right]
    % \lesssim
    \leq \mathbb{E}\left[\frac{\Delta_{\mu}}{2\varepsilon}\right] 
    \in \mathcal{O}(2^{-n/2} \frac{1}{\varepsilon}),
\end{equation}
up to the linear approximation error, where $1/\varepsilon$ is not an essential factor since it reflects the frequencies of the landscape fluctuation rather than magnitudes, similar to the role of the factor $\tr(\Omega_{\mu}^2)$ in the expression of $\operatorname{Var}[\partial_\mu C]$~\cite{McClean2018}.
% not involving a factor like $\tr(\Omega_{\mu}^2)$
% Note that vanishing gradients together with the parameter-shift rule could not deduce Theorem~1 since the maximization and minimization are performed before taking the expectation, i.e. for each sample. One can imagine an ensemble where Delta = const. but for each pair of points the difference is averaged to be zero.

% Apart from indicating vanishing gradients, Theorem~\ref{theorem:main-theorem} also implies vanishing the cost function difference $C(\bm{\theta}')-C(\bm{\theta})$ when $\{\theta_{\mu}\}_{\mu=1}^M$ enter in local quantum gates independently in the gradient-free optimization. 
Then consider the gradient-free optimization. The basis for a gradient-free optimizer to update parameters are cost function differences. For the cost function difference between any two parameter points $\bm{\theta}'$ and $\bm{\theta}$, Theorem~\ref{theorem:main-theorem} leads to
\begin{equation}
    \begin{aligned}
        & \mathbb{E} \left[ \left| C(\bm{\theta}')-C(\bm{\theta}) \right| \right] \\
        & \leq \mathbb{E} \left[ \sum_{\mu=1}^M \left| C\left( \bm{\theta}^{(\mu)} \right) - C\left( \bm{\theta}^{(\mu - 1)} \right) \right| \right] \\
        & \leq \sum_{\mu=1}^M \mathbb{E} \left[ \left| \Delta_\mu \right| \right] \in \mathcal{O}(M 2^{-n/2}),
    \end{aligned}
\end{equation}
where $\bm{\theta}^{(\mu)} = \bm{\theta} + \sum_{\nu=1}^{\mu} \left( \theta'_\nu - \theta_\nu \right)\bm{{\rm e}}_\nu$ for $\mu=1,...,M$ and $\bm{\theta}^{(\mu)} = \bm{\theta}$ for $\mu=0$. Thus, as long as the number of parameters satisfies $M\in\mathcal{O}(poly(n))$, the cost function difference between any two points vanish exponentially in the number of qubits with a high probability, demanding an exponential precision to make progress in the gradient-free optimization.

Furthermore, Theorem~\ref{theorem:main-theorem} goes beyond vanishing gradients and vanishing differences between two fixed points. The exponentially vanishing quantity claimed by Theorem~\ref{theorem:main-theorem} is the variation range of the cost function in the \textit{whole parameter subspace} corresponding to a local unitary, e.g., the subspace of the $3$ Euler angles in a single-qubit rotation gate from $\mathcal{SU}(2)$, or the subspace of the $15$ parameters in a two-qubit rotation gate from $\mathcal{SU}(4)$, etc. Therefore, Theorem~\ref{theorem:main-theorem} can be regarded as a fundamental limitation on optimization in VQAs and a essential mechanism behind barren plateaus.

\vspace{2mm}
\noindent\textbf{Application on representative VQAs.} 
To better illustrate the meaning of our findings in practice, we further investigate the applications of Theorem~\ref{theorem:main-theorem} on three representative VQAs, including the variational quantum eigensolver (VQE), quantum autoencoder, and quantum state learning. The corresponding numerical simulation results are summarized in Fig. \ref{fig:simulation}.

%%%%%%%%%%%%%%%%%%%%%%%%%%%%%%%%%%%%%%%%%%%%%%%%%%%%%%%%%%%%%
\paragraph{Application on VQE.}
The variational quantum eigensolver is the most famous VQA with the goal to prepare the ground state of a given Hamiltonian $\hat{H}$ of a physical system~\cite{Peruzzo2014}. The cost function is naturally chosen to be the expectation of the Hamiltonian with respect to an ansatz state $\mathbf{U}\ket{0}$, i.e.
\begin{equation}
    C_{\rm VQE}(\mathbf{U}) = \braandket{0}{ \mathbf{U}^\dagger \hat{H} \mathbf{U} }{0}. 
\end{equation}
For most physical models with local interactions, the spectral width is proportional to the system size, i.e., $w(\hat{H})\in\mathcal{O}(n)$. Hence from Theorem~\ref{theorem:main-theorem} we know that $\Delta_{\rm VQE}(V_1,V_2)$ vanishes exponentially with a high probability for random circuits forming $2$-designs. For common repeated-layer-type ansatzes, e.g., the hardware-efficient ansatzes~\cite{Kandala2017}, linear depth $\mathcal{O}(n)$ is enough to make a randomly initialized circuit to be a sample from an approximate $2$-design ensemble~\cite{McClean2018,Harrow2009,Brandao2016}. We conduct numerical simulations for the variation range of the VQE cost function $\Delta_{{\rm VQE}}$ using the $1$-dimensional spin-$1/2$ antiferromagnetic Heisenberg model
\begin{equation}\label{Eq:Heisenberg-model}
    \hat{H}=\sum_{i=1}^{n}\left( X_{i} X_{i+1} + Y_{i} Y_{i+1} + Z_{i} Z_{i+1} \right),
\end{equation}
with periodic boundary condition, as shown in Fig. \ref{fig:simulation}(a).

\begin{figure}[tp]
    \includegraphics[width=\columnwidth]{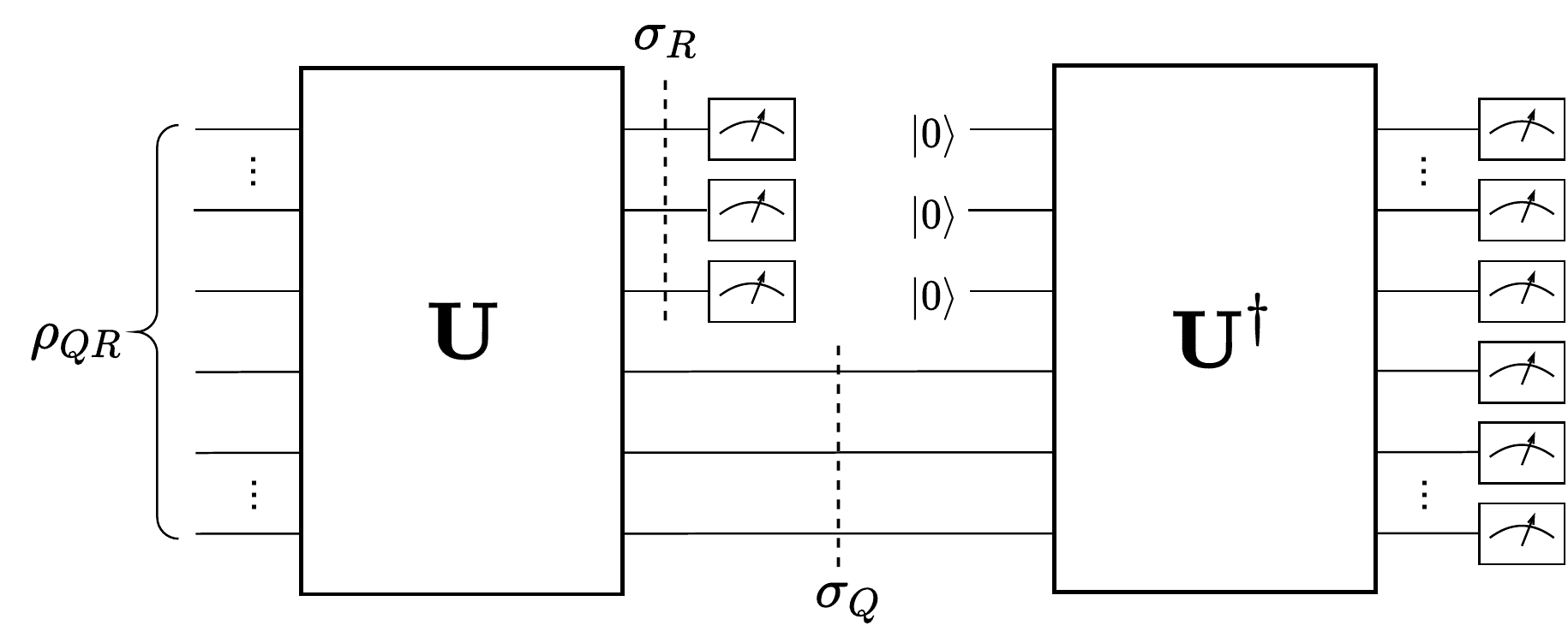}
    \caption{\textbf{Circuit setting of the quantum autoencoder.} $\rho_{QR}$ is the given state to be compressed and $\sigma_Q$ is the compressed state through the encoder $\mathbf{U}$. The quantum autoencoder aims to train $\mathbf{U}$ such that $\rho_{QR}$ can be reconstructed from $\sigma_Q$ with a high fidelity through the decoder $\mathbf{U}^\dagger$ combined with an ancilla zero state $\ketbra{0}{0}_R$. $\sigma_R$ denotes the state of the discarded part after compression.}
    \label{fig:autoencoder}
\end{figure}

\paragraph{Application on quantum autoencoder.}
The quantum autoencoder (QAE) is an approach for quantum data compression~\cite{Romero2017,Cao2020}. As shown in Fig.~\ref{fig:autoencoder}, a quantum circuit $\mathbf{U}$ is trained as an encoder to compress a given state $\rho_{QR}$ on a bipartite system $QR$ into a reduced state $\sigma_{Q} = \tr_{R}(\mathbf{U} \rho_{QR} \mathbf{U}^{\dagger})$ on subsystem $Q$, such that $\rho_{QR}$ can be reproduced from $\sigma_{Q}$ by the decoder isometry $\bra{0}_R\mathbf{U}^\dagger$ with a high fidelity. According to the monotonicity of the fidelity under partial trace, an easy-to-measure cost function could be reduced from the fidelity between $\rho_{QR}$ and the reconstructed state as
\begin{equation}\label{Eq:cost_QAE}
    C_{\rm QAE}(\mathbf{U}) := 1 - \tr\left((\ketbra{0}{0}_R\otimes I_Q) \mathbf{U} \rho_{QR} \mathbf{U}^\dagger \right).
\end{equation}
where the second term is exactly the fidelity between the state of the discarded part $\sigma_R = \tr_{Q}(\mathbf{U} \rho_{QR} \mathbf{U}^{\dagger})$ and the zero state $\ket{0}_R$ on subsystem $R$. The spectral width for the QAE cost function (\ref{Eq:cost_QAE}) is $w(H_{{\rm QAE}}) = 1$ with $H_{{\rm QAE}}=I_{QR}-\proj{0}_R\otimes I_{Q}$. Thus again from Theorem~\ref{theorem:main-theorem} we know that $\Delta_{\rm QAE}(V_1,V_2)$ vanishes exponentially in the number of qubits, specifically with the scaling $\mathcal{O}(2^{-n/2})$ as shown in Fig.~\ref{fig:simulation}(b). 

\begin{figure*}[ht]
    \includegraphics[width=\textwidth]{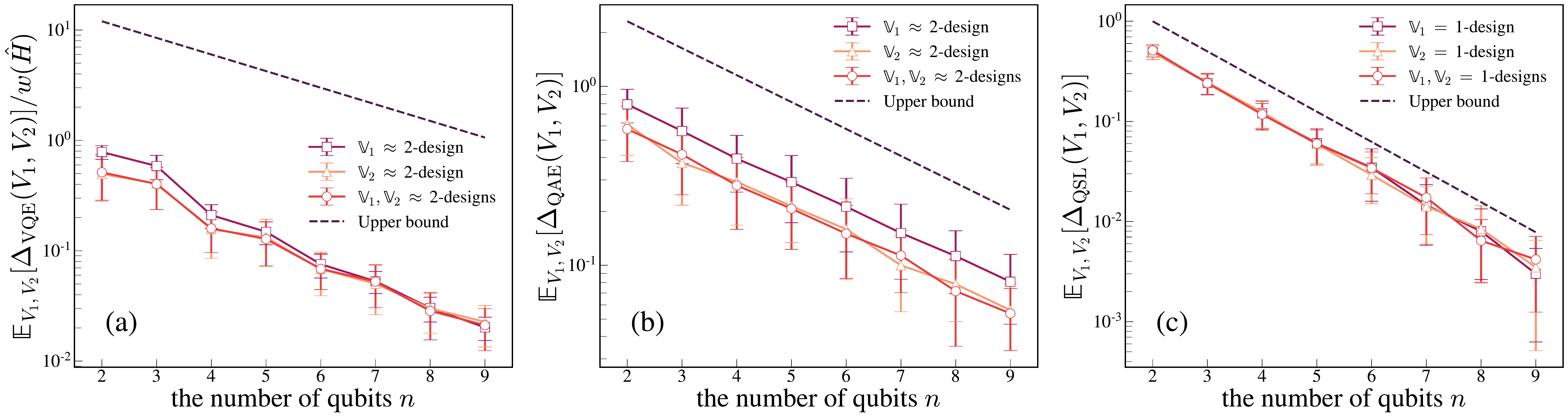}
    \caption{\textbf{Exponentially vanishing variation range of the cost function via varying a local unitary.} The data points represent the sample averages of the cost variation range $\Delta_{H,\rho}$ via varying a single-qubit unitary over the spectral width $w(H)$ as a function of the number of qubits on semi-log plots. Panel (a) and (b) correspond to the VQE with the $1$-dimensional Heisenberg model and the quantum autoencoder with one qubit discarded, respectively, where the error bars represent the standard deviations over samples. Panel (c) corresponds to the quantum state learning with the cost function being the fidelity with the zero state. Different legends stand for $\mathbb{V}_1$, $\mathbb{V}_2$ or both being approximate $2$-designs in (a), (b) and $1$-designs in (c). The dashed lines depict our theoretical upper bounds for the three tasks where the scaling exponents shows a good coincidence with the experimental results.}
    \label{fig:simulation}
\end{figure*}
%%%%%%%%%%%%%%%%%%%%%%%%%%%%%%%%%%%%%%%%%%%%%%
\paragraph{Application on quantum state learning.}
The fidelity between pure states is a special case of the VQA cost function in (\ref{Eq:VQA_cost_def}) with a low-rank observable. Many useful VQA applications make use of the fidelity as their cost functions~\cite{Lee2018a,Shirakawa2021,Bravo-Prieto2019}. Here we uniformly call them quantum state learning (QSL) tasks. Denote the input state as $\ket{\psi}$ and the target state as $\ket{\phi}$. The QSL cost function can be written as
\begin{equation}\label{Eq:loss_QSL}
    C_{\rm QSL}(\mathbf{U}) = 1 - \left|\braandket{\phi}{\mathbf{U}}{\psi}\right|^2,
\end{equation}
Theorem~\ref{theorem:main-theorem} certainly can be applied here with $H_{{\rm QSL}} = I - \ketbra{\phi}{\phi}$ and $w(H_{\rm QSL})=1$. But here we can provide a tighter and stronger bound for the variation range $\Delta_{QSL}$ in this special case as Proposition~\ref{prop:state-learning}, which generally holds for the Bures fidelity $F(\rho,\sigma)=\left(\tr\sqrt{\rho^{1/2}\sigma\rho^{1/2}}\right)^2$ between arbitrary density matrices $\rho$ and $\sigma$. A generalized version of the QSL cost function is
\begin{equation}\label{Eq:qsl_any_density_matirx}
    C_{\rm QSL}(\mathbf{U}) = 1 -  F\left(\mathbf{U}\rho \mathbf{U}^\dagger, \sigma\right).
\end{equation}
The proof of Proposition~\ref{prop:state-learning} is sketched in the Methods and detailed in Supplementary Note~\textcolor{blue}{3}.
\begin{proposition}\label{prop:state-learning}
Suppose $\mathbb{V}_1, \mathbb{V}_2$ are ensembles from which $V_1, V_2$ are sampled, respectively. If either $\mathbb{V}_1$ or $\mathbb{V}_2$, or both form unitary $1$-designs, then for arbitrary $\rho$ and $\sigma$, the following inequality holds
\begin{equation}
    \mathbb{E}_{V_1,V_2} \left[ \Delta_{\rm QSL}(V_1,V_2) \right] \leq \frac{1}{2^{n-2m}}.
\end{equation}
\end{proposition}

Compared with Theorem~\ref{theorem:main-theorem}, the bound $\mathcal{O}(2^{-n})$ becomes tighter and the demanded randomness becomes weaker in this special case. Note that a unitary $2$-design is always a $1$-design by definition and the depth of a random circuit being a $1$-design is much shallower, e.g., a single layer of $\mathcal{SU}(2)$ elements on each qubit parameterized by $3$ Euler angles is enough to form a $1$-design ensemble. Like in (\ref{Eq:VarV1V2_delta}) and (\ref{Eq:PrV1V2_delta}), the variance and the probability that $\Delta_{{\rm QSL}}$ deviates from zero also vanish exponentially, but only require random circuits forming unitary $1$-designs. Moreover, still with $1$-designs, Proposition~\ref{prop:state-learning} implies exponentially vanishing cost gradients and cost differences in the same way as Theorem~\ref{theorem:main-theorem}, which may be considered as the underlying mechanism behind the severe barren plateaus for global cost functions even with shallow quantum circuits~\cite{Cerezo2021}.

%%%%%%%%%%%%%%%%%%%%%%%%%%%%%%%%%%%%%%%%%%%%%%%%%%%%%%%
\vspace{2mm}

\noindent\textbf{Numerical simulations of experiments.} 
To show the validity of our results, we carry out numerical simulations of experiments on the three tasks discussed above via Paddle Quantum~\cite{Paddle-Quantum} on the PaddlePaddle Deep Learning Platform~\cite{Ma2019}.

\setcounter{paragraph}{0}
\paragraph{Circuit setting.} 
We first introduce the circuit settings used in our simulations. We consider subsystem $A$ only containing a single qubit, namely $m=1$, and parameterize the local unitary $U_A\in\mathcal{U}(2)$ with $3$ Euler angles up to a global phase, i.e., $U_A(\phi, \theta, \alpha) = R_z(\phi) R_y(\theta) R_z(\alpha)$, where $R_y$ and $R_z$ are single-qubit rotation gates with generators being $Y$ and $Z$ Pauli matrices. To construct random circuits forming $2$-designs as $V_1$ or $V_2$ used in the VQE and QAE examples, we employ the following hardware-efficient ansatz as in~\cite{McClean2018} for comparison. 
\begin{equation}\label{Cir:hardware_efficient}
    \begin{array}{c}
        \centering
        \Qcircuit @C=1em @R=0.3em {
        & \gate{R_y(\frac{\pi}{4})} & \gate{R_{P_{1,1}}(\theta_{1,1})} & \ctrl{1} & \qw & \qw & \qw & \qw &\cdots\\
        & \gate{R_y(\frac{\pi}{4})} & \gate{R_{P_{1,2}}(\theta_{1,2})} & \control \qw & \ctrl{1} & \qw & \qw &\qw & \cdots\\
        & \gate{R_y(\frac{\pi}{4})} & \gate{R_{P_{1,3}}(\theta_{1,3})} & \qw & \control \qw & \ctrl{2} & \qw & \qw & \cdots\\
        \\
        & \cdots & \cdots &  & \cdots &  & & & \cdots\\
        \\
        & \gate{R_y(\frac{\pi}{4})} & \gate{R_{P_{1,n}}(\theta_{1,n})} & \qw & \qw & \qw & \ctrl{-2} \qw & \qw & \cdots & \dstick{\times 10n} \gategroup{1}{3}{7}{8}{.7em}{--}  \\
        } 
    \end{array}
\end{equation}
A single layer of $R_y(\pi/4)=\exp(-i Y\pi/8)$ gates are laid at the very beginning of the circuit to make the three rotation axes have equal status, then followed by $10\times n$ repeated layers. Each layer consists $n$ single-qubit rotation gates $R_P(\theta)$ on each qubit together with $n-1$ controlled phase gates between nearest neighboring qubits aligned as 1-dimensional array, where the rotation axes $P\in \{x,y,z\}$ is chosen with uniform probability and $\theta \in [0, 2\pi)$ is also chosen uniformly. A such random circuit with $\mathcal{O}(n)$ repeated layers could be considered as an approximate $2$-design (here we employ $10\times n$)~\cite{McClean2018,Harrow2009,Brandao2016}. Experimental results with different numbers of layers are also presented in Supplementary Note~\textcolor{blue}{4} to show how the expectation of the cost variation range $\Delta_{H,\rho}$ vanishes with the circuit depth. To construct random circuits forming $1$-designs as $V_1$ or $V_2$ used in the QSL example, we just replace the repeated layers above by a single layer of $\mathcal{SU}(2)$ elements $R_z(\phi) R_y(\theta) R_z(\alpha)$ on each qubit with $\phi,\theta,\alpha\in[0,2\pi)$ are chosen with uniform probability, which is enough to form a unitary $1$-design ensemble.

\paragraph{Implementation procedure.}
To implement the maximization and minimization in the definition of $\Delta_{H,\rho}(V_1,V_2)$ with respect to $U_A$, we employ the Adam optimizer to update $U_A$ iteratively until convergence for each sample of $V_1,V_2$, and hence we get $\Delta_{H,\rho}(V_1,V_2)$. We consider the converged value as a good estimation of the true values of $\max_{U_A}C$ and $\min_{U_A}C$ with a tolerable error, at least for circuits with a small number of qubits ($\leq 10$) and a modest depth ($\leq 10\times n$). We repeat this procedure for different number of qubits and different statistics of $\mathbb{V}_1$ and $\mathbb{V}_2$, i.e., $\mathbb{V}_1$ or $\mathbb{V}_2$ being a $2$-design ($1$-design) while the other being identity.
%i.e., $U_A$ locating before a $2$-design, behind a $2$-design and between two $2$-designs, or $1$-design.

\paragraph{Numerical results.}
We summarize the simulation results for the three examples in the three semi-log plots in Fig.~\ref{fig:simulation}, respectively. Note that the slopes of the lines imply the rates of exponential decay. The data points represent the sample averages of the cost variation range $\Delta_{H,\rho}$ via varying $U_A$ over $w(H)$ and the error bars represent the standard deviations over samples. We specially rescale the error bar in the QSL example as a quarter of the standard deviation for better presentation on semi-log plots. One can see that in all the cases, the expectations of $\Delta_{H,\rho}(V_1,V_2)$ vanish exponentially in the number of qubits. The data lines are almost parallel to the dashed lines depicting the theoretical upper bounds. That is to say, the scaling behaviors, $\mathcal{O}(2^{-n/2})$ in Fig.~\ref{fig:simulation}(a), (b) and $\mathcal{O}(2^{-n})$ in (c), almost coincide with the predictions from Theorem~\ref{theorem:main-theorem} and Proposition~\ref{prop:state-learning}. The VQE example actually has a fitted slope a little steeper than $-0.5$ with a model-dependent odd-even oscillation. A detailed derivation can be found in Supplementary Note~\textcolor{blue}{2} for the tighter task-dependent upper bounds used in Fig.~\ref{fig:simulation}(a) and (b).

%%%%%%%%%%%%%%%%%%%%%%%%%%%%%%%%%%%%%%%%%%%%%%%%%%%%%%%%%%%%%%%%%%%%%%%%%%%%%%%
\vspace{3mm}
\noindent\textbf{Discussion}\\
In this work, we have shown that the maximal possible influence of a local unitary within a random quantum circuit on the cost function vanishes exponentially in the number of qubits with a high probability. We remark that the randomness required is just a $2$-design for the general VQA cost function in Theorem~\ref{theorem:main-theorem} and a $1$-design for the fidelity-type cost function in Proposition~\ref{prop:state-learning}, in spite that the integrand $\Delta_{H,\rho}(V_1,V_2)$ is not necessarily a polynomial of degree at most $2$ or $1$ in the entries of $V_1$ and $V_2$ due to the maximization and minimization operations. The exponential bound still works non-trivially for global unitaries satisfying the parameter-shift rule.

Significantly, it is worth noticing that the exponentially vanishing quantity we claimed is the \textit{entire} variation range of the cost function in the \textit{whole parameter subspace} corresponding to an arbitrary local unitary, e.g., the subspace of the $3$ Euler angles in a single-qubit rotation gate from $\mathcal{SU}(2)$, or the subspace of the $15$ parameters in a two-qubit rotation gate from $\mathcal{SU}(4)$, etc, which is in sharp contrast with the gradient analyses on the VQA training landscapes.

Theorem~\ref{theorem:main-theorem} together with Proposition~\ref{prop:state-learning}, as new rigorous scaling results for generic VQAs beyond gradients, can unify the restrictions on gradient-based and gradient-free optimizations in a natural way and be regarded as the underlying mechanism behind the barren plateau phenomenon. Additionally, the results are completely independent with quantum circuit details such as gate parameterization. Therefore, a fundamental limitation is unravelled on optimization in VQAs, which can serve as a guidance for designing better training strategies to improve the scalability of VQAs. A direct consequence is that, since our results hold regardless of optimizer choices, the gate-by-gate optimization strategy is ineffective no matter what optimizers are utilized. Reparameterization within local unitaries is also unhelpful. Further excluding the strategies which have already been ruled out by the gradient analyses~\cite{McClean2018,Arrasmith2020,Cerezo2021c}, the hopeful strategies may exist in proper initialization~\cite{Grant2019}, pre-training including adaptive methods~\cite{Verdon2019,Grimsley2019,Zhang2021,Skolik2021,Grimsley2022}, circuit architectures~\cite{Pesah2021,Liu2022a} and cost function choices~\cite{Cerezo2021,Kieferova2021}, etc.

\vspace{3mm}
\noindent\textbf{Methods}\\
Here we give a sketch of the proof of Theorem~\ref{theorem:main-theorem} and Proposition~\ref{prop:state-learning} and remain the details in the Supplementary Note~\textcolor{blue}{2} and \textcolor{blue}{3}, respectively. 

\vspace{3mm}
\noindent\textbf{Sketch of the proof of Theorem~\ref{theorem:main-theorem}}. 
Without loss of generality, we assume that $H$ is traceless since (\ref{Eq:main-theorem}) is invariant if $H$ is added by a homothety $H\rightarrow H + c I$, $c\in\mathbb{R}$. Moreover, we only need to find a upper bound for the maximization term in $\Delta_{H,\rho}$ since the bound for the minimization term could be obtained similarly by replacing $H$ with $-H$. 

On the one hand, if $\mathbb{V}_1$ is a $2$-design, we can first take the expectation over $\mathbb{V}_1$, i.e.,
\begin{equation}\label{Eq:V_1-2design-EdeltaC}
    \mathbb{E}_{V_1} \max_{U_A} \left[ \tr\left( H \mathbf{U} \rho \mathbf{U}^\dagger \right) \right].
\end{equation}
where $\mathbf{U}=V_2(U_A\otimes I_B)V_1$. To find an upper bound for this expectation, we expand the traceless Hermitian operator $\tilde{H} = V_2^\dagger H V_2$ in terms of Pauli strings  $\hat{\sigma}_j^A$ on subsystem $A$ as 
\begin{equation}\label{Eq:H_decomposition}
     \tilde{H} = \tr_B(\tilde{H}) \otimes \frac{I_B}{2^{n-m}} + \frac{I_A}{2^m} \otimes \tr_A(\tilde{H}) + \sum_{j=1}^{4^{m}-1} \hat{\sigma}_j^A\otimes O_j^B,
\end{equation}
where the first term only acts on subsystem $A$ non-trivially, the second term only acts on subsystem $B$ non-trivially and the last terms act on $A$ and $B$ both non-trivially. $O_j^B$ represents the sum of Pauli strings on $B$ corresponding to $\hat{\sigma}_j^A$. Note that every decomposed term has a bipartite tensor product structure. We bound (\ref{Eq:V_1-2design-bound}) by taking the maximization and expectation for each term respectively and add them together at last. For each term, we use Hölder's inequality to extract $U_A$ out and bound the remaining part with specific calculations of $2$-design element-wise integrals (see Supplementary Note~\textcolor{blue}{1}). Consequently, we arrive at
\begin{equation}\label{Eq:V_1-2design-bound}
    \mathbb{E}_{V_1} \max_{U_A} \left[ \tr\left( H \mathbf{U} \rho \mathbf{U}^\dagger \right) \right]
    \leq \frac{w(\tilde{H})}{ 2^{n/2-3m-1}},
\end{equation}
Since $w(\tilde{H}) = w(V_2^\dagger H V_2) = w(H)$ for arbitrary $V_2$, the bound in (\ref{Eq:V_1-2design-bound}) is valid even after taking expectation over $\mathbb{V}_2$ no matter what the ensemble $\mathbb{V}_2$ contains. On the other hand, if $\mathbb{V}_2$ is a $2$-design, we can perform the decomposition for $V_1\rho V_1^\dagger$ and evaluate the upper bound in a similar spirit.

\vspace{3mm}
\noindent\textbf{Sketch of the proof of Proposition~\ref{prop:state-learning}}. Suppose $\mathbb{V}_1$ is a $1$-design and absorb $V_2$ into the definition of $\sigma$ as $V_2\sigma V_2^\dagger$. The case of $\mathbb{V}_2$ being a $1$-design can be proved similarly. 

Due to the non-negativity of the fidelity, the cost variation range $\Delta_{{\rm QSL}}$ is not larger than the maximization term in its definition (\ref{Eq:def_delta_V1V2}). Then according to the monotonicity of the fidelity function under partial trace, the fidelity between two quantum states can be upper bounded by the fidelity between the corresponding reduced states on a subsystem, i.e.
\begin{equation}\label{Eq:F_monotonicity}
    F((U_A\otimes I_B)V_1\rho V_1^\dagger(U_A\otimes I_B)^\dagger,\sigma) \leq F(\rho_{B},\sigma_{B}),
\end{equation}
where $\rho_{B}=\tr_A(V_1\rho V_1^\dagger)$ and $\sigma_{B}=\tr_A(\sigma)$. (\ref{Eq:F_monotonicity}) holds for arbitrary $U_A$ and hence holds for the case of maximizing with respect to $U_A$. Then by definition the Bures fidelity is the squared Schatten $1$-norm of the square root of the product of two density matrices, which can be upper bounded by the squared Schatten $2$-norm times its rank. This leads to
\begin{equation}
    F(\rho_{B},\sigma_{B}) \leq \operatorname{rank}(\rho_{B}\sigma_{B}) \tr(\rho_{B}\sigma_{B}).
\end{equation}
The rank is not larger than the Hilbert space dimension of subsystem $A$ and the expectation of $\tr(\rho_{B}\sigma_{B})$ over $\mathbb{V}_1$ can be exactly calculated by use of the condition of unitary $1$-design, where a exponentially small factor $2^{-n}$ would emerge.

\vspace{3mm}
\noindent\textbf{Acknowledgements.}\\
H. Z. and C. Z. contributed equally to this work. Part of this work was done when H. Z., C. Z., and G. L. were research interns at Baidu Research. We would like to thank Runyao Duan and Zihe Wang for helpful discussions.

%%%%%%%%%%%%%%%%%%%%%%%%%%%%%%%%%%%%%%%%%%%%%%%%%%%%%%%%%%%%%%%%%%%%%%%%%
% Bibliography
%%%%%%%%%%%%%%%%%%%%%%%%%%%%%%%%%%%%%%%%%%%%%%%%%%%%%%%%%%%%%%%%%%%%%%%%%%%%%%
%\bibliographystyle{apsrev4-1}
\bibliography{autoref2}

%%%%%%%%% SUPPLEMENTAL MATERIAL %%%%%%%%%

\clearpage

\appendix
\setcounter{subsection}{0}
\setcounter{table}{0}
\setcounter{figure}{0}

\vspace{2cm}
\onecolumngrid
\vspace{2cm}

\begin{center}
\large{Supplementary Information for \\ \textbf{
{Fundamental limitations on optimization in variational quantum algorithms}}}
\end{center}

\renewcommand{\theequation}{S\arabic{equation}}
% \numberwithin{equation}{section}
\renewcommand{\thesubsection}{\normalsize{Supplementary Note \arabic{subsection}}}
\renewcommand{\theproposition}{S\arabic{proposition}}
\renewcommand{\thedefinition}{S\arabic{definition}}
\renewcommand{\thefigure}{S\arabic{figure}}
\setcounter{equation}{0}
\setcounter{table}{0}
\setcounter{section}{0}
\setcounter{proposition}{0}
\setcounter{definition}{0}
\setcounter{figure}{0}

In this Supplementary Information, we present detailed proofs of the theorems, propositions in the manuscript ``Fundamental limitations on optimization in variational quantum algorithms''. In \ref{supplementary:preliminaries}, we review and derive several useful identities about integrals over unitary groups and some fundamental inequalities in order to make our proofs more self-contained. In \ref{supplementary:proof_main_theorem}, we give a detailed proof of the Theorem~\textcolor{blue}{1} in the manuscript of a more general version that holds for arbitrary dimensions instead of only qubit systems. Then in \ref{supplementary:proof_prop_qsl}, we provide the proof for Proposition \textcolor{blue}{2}. Finally in \ref{supplementary:varying_layers}, we display some numerical simulation results on the variation range of the cost function with different numbers of circuit layers.

\subsection{Preliminaries}\label{supplementary:preliminaries}

We start from the definition of a unitary $t$-design~\cite{Dankert2009}. Consider a ensemble $\mathbb{V}$ of unitaries $V$ on a $d$-dimensional Hilbert space, and denote $P_{t,t}(V)$ as an arbitrary polynomial of degree at most $t$ in the entries of $V$ and at most $t$ in those of $V^\dagger$. Then $\mathbb{V}$ is a unitary $t$-design if
\begin{equation}\label{Eq:definiton_design}
    \frac{1}{|\mathbb{V}|}\sum_{V\in\mathbb{V}} P_{t,t}(V) = \int_{\mathcal{U}(d)} d\mu(V) P_{t,t}(V),
\end{equation}
where $|\mathbb{V}|$ is the size of the set $\mathbb{V}$, $\mathcal{U}(d)$ is the unitary group of degree $d$ and $d\mu(V)$ is the Haar measure on $\mathcal{U}(d)$. Namely, $P_{t,t}(V)$ averaging over the $t$-design $\mathbb{V}$ will yield exactly the same result as averaging over the entire unitary group $\mathcal{U}(d)$.
% \begin{definition} \label{lemma:t-design trace form}
% For all finite set $X \subset U(d)$, we have
% \begin{align}
%     \sum_{U,V\in X}w(U)w(V)|\tr[U^{\dagger}V]|^{2t} \geq \int_{U(d)}|\tr[U]|^{2t} d U
% \end{align}
% With equality holds if and only if $X$ is a $t$-design
% \end{definition}
% \begin{lemma}\label{n_tensor_2design_append}
% If $X\subset U(d)$ is a unitary $t$-design, then for any $M\in U(d)$, $M X = \{MU: U\in X\}$ is also a $t$-design
% \end{lemma}
% \begin{proof}
% Since $X$ is a t-design, we have
% \begin{align}
%   \frac{1}{|X|^2}\sum_{U,V\in X}|\tr[U^{\dagger}V]|^{2t} = \int_{U(d)}|\tr[U]|^{2t}d U 
% \end{align}
% Considering the sum over $M X$,
% \begin{align}
%     &\frac{1}{|MX|^2}\sum_{\hat{U},\hat{V}\in MX}|\tr[\hat{U}^{\dagger}\hat{V}]|^{2t}=  \frac{1}{|X|^2}\sum_{U,V\in X}|\tr[(MU)^{\dagger}(MV)]|^{2t}\\
%     &= \frac{1}{|X|^2}\sum_{U,V\in X}|\tr[U^{\dagger}M^{\dagger}MV]|^{2t}
%     = \frac{1}{|X|^2}\sum_{U,V\in X}|\tr[U^{\dagger}V]|^{2t}
%     = \int_{U(d)}|\tr[U]|^{2t}dU.
% \end{align}
% According to the lemma \ref{lemma:t-design trace form}, $MX$ is also a $t$-design.
% \end{proof}
Fortunately, these integrals over polynomials can be analytically solved and expressed into closed forms. For example, the following element-wise identities hold for the first two moments~\cite{Collins2006,Puchaa2017}
\begin{subequations}
    \begin{align}
        \int_{\mathcal{U}(d)} d \mu(V) v_{i, j} v_{i', j'}^{*} 
        & = \frac{\delta_{i, i'} \delta_{j, j'}}{d}, 
        \label{Eq:1-moment_element-wise}\\
        \int_{\mathcal{U}(d)} d \mu(V) 
        v_{i_{1}, j_{1}} 
        v_{i_{2}, j_{2}} 
        v_{i_{1}^{\prime}, j_{1}^{\prime}}^{*} 
        v_{i_{2}^{\prime}, j_{2}^{\prime}}^{*}
        &=\frac{1}{d^{2}-1}\left(
        \delta_{i_{1}, i_{1}^{\prime}} 
        \delta_{i_{2}, i_{2}^{\prime}} 
        \delta_{j_{1}, j_{1}^{\prime}} 
        \delta_{j_{2}, j_{2}^{\prime}}
        +
        \delta_{i_{1}, i_{2}^{\prime}} 
        \delta_{i_{2}, i_{1}^{\prime}} 
        \delta_{j_{1}, j_{2}^{\prime}} 
        \delta_{j_{2}, j_{1}^{\prime}} \right)\notag \\
        &-\frac{1}{d\left(d^{2}-1\right)}\left(
        \delta_{i_{1}, i_{1}^{\prime}} 
        \delta_{i_{2}, i_{2}^{\prime}} 
        \delta_{j_{1}, j_{2}^{\prime}} 
        \delta_{j_{2}, j_{1}^{\prime}}
        +
        \delta_{i_{1}, i_{2}^{\prime}} 
        \delta_{i_{2}, i_{1}^{\prime}} 
        \delta_{j_{1}, j_{1}^{\prime}} 
        \delta_{j_{2}, j_{2}^{\prime}} \right),
        \label{Eq:2-moment_element-wise}
    \end{align}
\end{subequations}
where $v_{i,j}$ and $v^*_{i',j'}$ denote the entries of $V$ and $V^*$, respectively, and $\delta_{i,j}$ denotes the Kronecker delta. For practical purposes, these element-wise identities need to be transformed into various matrix forms, during which one will encounter many contraction operations. Here we take advantage of tensor network notations to deal with the contraction operations. For example, if we arrange the indices like
\begin{equation}
    v_{i,j} = \left(
    \diagram{
	    % radius
	    \def\r{0.5};
	    % space
	    \def\spacex{1.4};
	    \def\spacey{1.4};
	    % x-coordinate
	    \def\xc{0};
	    % y-coordinate
	    \def\yc{0};
	    % coordinate of circle
	    \coordinate (pc) at (\xc,\yc);
	    % horizontal line
	    \draw (\xc-2*\r,\yc) -- (\xc-\r,\yc);
	    \draw (\xc+\r,\yc) -- (\xc+2*\r,\yc);
	    % tensor circle
	    \draw[ten, shift=(pc)] (-\r,-\r) rectangle (\r,\r);
		\node at (pc) {\footnotesize $V$};
		\node at (\xc-3*\r,\yc) {\footnotesize $i$};
		\node at (\xc+3*\r,\yc) {\footnotesize $j$};
	}\right),~~
	v^*_{i',j'} = \left(
    \diagram{
	    % radius
	    \def\r{0.5};
	    % space
	    \def\spacex{1.4};
	    \def\spacey{1.4};
	    % x-coordinate
	    \def\xc{0};
	    % y-coordinate
	    \def\yc{0};
	    % coordinate of circle
	    \coordinate (pc) at (\xc,\yc);
	    % horizontal line
	    \draw (\xc-2*\r,\yc) -- (\xc-\r,\yc);
	    \draw (\xc+\r,\yc) -- (\xc+2*\r,\yc);
	    % tensor circle
	    \draw[ten, shift=(pc)] (-\r,-\r) rectangle (\r,\r);
		\node at (pc) {\footnotesize $V^*$};
		\node at (\xc-3*\r,\yc) {\footnotesize $i'$};
		\node at (\xc+3*\r,\yc) {\footnotesize $j'$};
	}\right),
\end{equation}
(\ref{Eq:1-moment_element-wise}) could be represented as the following diagram
\begin{equation}\label{Eq:1-moment_tn}
    \int_{\mathcal{U}(d)} d \mu(V) \left(
    \diagram{
	    % radius
	    \def\r{0.5};
	    % space
	    \def\spacex{1.4};
	    \def\spacey{1.4};
	    % x-coordinate
	    \def\xc{0};
	    % y-coordinate
	    \def\yc{0};
	    \def\yu{1*\spacey};
	    % coordinate of circle
	    \coordinate (pc) at (\xc,\yc);
	    \coordinate (puc) at (\xc,\yu);
	    % horizontal line
	    \draw (\xc-2*\r,\yc) -- (\xc-\r,\yc);
	    \draw (\xc+\r,\yc) -- (\xc+2*\r,\yc);
	    \draw (\xc-2*\r,\yu) -- (\xc-\r,\yu);
	    \draw (\xc+\r,\yu) -- (\xc+2*\r,\yu);
	    % tensor circle
	    \draw[ten, shift=(puc)] (-\r,-\r) rectangle (\r,\r);
		\node at (puc) {\scriptsize $V$};
	    \draw[ten, shift=(pc)] (-\r,-\r) rectangle (\r,\r);
		\node at (pc) {\scriptsize $V^*$};
	}\right)
	=  \frac{1}{d} \left(
    \diagram{
    	% radius
	    \def\r{0.5};
	    % space
	    \def\spacex{1.4};
	    \def\spacey{1.4};
	    % x-coordinate
	    \def\xc{0};
	    % y-coordinate
	    \def\yc{0};
	    \def\yu{1*\spacey};
	    % coordinate of circle
	    \coordinate (pc) at (\xc,\yc);
	    \coordinate (puc) at (\xc,\yu);
        \draw (\xc-2*\r,\yc) .. controls (\xc,\yc) and (\xc,\yu) .. (\xc-2*\r,\yu);
        \draw (\xc+2*\r,\yc) .. controls (\xc,\yc) and (\xc,\yu) .. (\xc+2*\r,\yu);
    }\right),
\end{equation}
where the arcs on the right hand side of (\ref{Eq:1-moment_tn}) represent identity matrices, i.e. the Kronecker delta $\delta_{i,i'}$ and $\delta_{j,j'}$. As a simple instance, we first prove Lemma \ref{lemma:integral_VAV=trAd} using (\ref{Eq:1-moment_tn}).
\begin{lemma}\label{lemma:integral_VAV=trAd}
For an arbitrary linear operator $A$ on the $d$-dimensional Hilbert space, the following equality holds
\begin{equation}\label{Eq:integral_VAV=trAd}
    \int_{\mathcal{U}(d)} V A V^\dagger d \mu(V)
    = \frac{\tr(A)}{d} I,
\end{equation}
where $I$ is the identity operator on the $d$-dimensional Hilbert space.
\end{lemma}
\begin{proof}
By tensor network notations and (\ref{Eq:1-moment_tn}), we have
\begin{equation}
    \int_{\mathcal{U}(d)} V A V^\dagger d \mu(V) 
    = \int_{\mathcal{U}(d)} d \mu(V) \left(
    \diagram{
	    % radius
	    \def\r{0.5};
	    % space
	    \def\spacex{1.4};
	    \def\spacey{1.4};
	    % x-coordinate
	    \def\xc{0};
	    \def\xr{\xc+3*\r}
	    % y-coordinate
	    \def\yc{0};
	    \def\yu{1*\spacey};
	    % coordinate of circle
	    \coordinate (pc) at (\xc,\yc);
	    \coordinate (puc) at (\xc,\yu);
	    \coordinate (pur) at (\xr,\yu);
	    % horizontal line
	    \draw (\xc-2*\r,\yc) -- (\xc-\r,\yc);
	    \draw (\xc+\r,\yc) -- (\xr+\r,\yc);
	    \draw (\xc-2*\r,\yu) -- (\xc-\r,\yu);
	    \draw (\xc+\r,\yu) -- (\xr+\r,\yu);
	    % curve
	    \draw (\xr+\r,\yu) .. controls (\xr+3*\r,\yu) and (\xr+3*\r,\yc) .. (\xr+\r,\yc);
	    % tensor circle
	    \draw[ten, shift=(puc)] (-\r,-\r) rectangle (\r,\r);
		\node at (puc) {\scriptsize $V$};
	    \draw[ten, shift=(pc)] (-\r,-\r) rectangle (\r,\r);
		\node at (pc) {\scriptsize $V^*$};
	    \draw[ten, shift=(pur)] (-\r,-\r) rectangle (\r,\r);
		\node at (pur) {\scriptsize $A$};
	}\right)
	=  \frac{1}{d} \left(
    \diagram{
    	% radius
	    \def\r{0.5};
	    % space
	    \def\spacex{1.4};
	    \def\spacey{1.4};
	    % x-coordinate
	    \def\xc{0};
	    \def\xr{\xc+3*\r}
	    % y-coordinate
	    \def\yc{0};
	    \def\yu{1*\spacey};
	    % coordinate of circle
	    \coordinate (pc) at (\xc,\yc);
	    \coordinate (puc) at (\xc,\yu);
	    % line
	    \draw (\xr-\r,\yc) -- (\xr+\r,\yc);
	    % curve
        \draw (\xc-2*\r,\yc) .. controls (\xc,\yc) and (\xc,\yu) .. (\xc-2*\r,\yu);
        \draw (\xc+2*\r,\yc) .. controls (\xc,\yc) and (\xc,\yu) .. (\xc+2*\r,\yu);
        \draw (\xr+\r,\yu) .. controls (\xr+3*\r,\yu) and (\xr+3*\r,\yc) .. (\xr+\r,\yc);
        % node
        \draw[ten, shift=(pur)] (-\r,-\r) rectangle (\r,\r);
		\node at (pur) {\scriptsize $A$};
    }\right)
    = \frac{\tr(A)}{d} I,
\end{equation}
which is exactly the same with (\ref{Eq:integral_VAV=trAd}).
\end{proof}

Similarly, (\ref{Eq:2-moment_element-wise}) could be represented by tensor network notations as the following diagram
\begin{equation}\label{Eq:2-moment_element-wise_tn}
    \int_{\mathcal{U}(d)} d \mu(V) \left(
    \diagram{
	    % radius
	    \def\r{0.5};
	    % space
	    \def\spacex{1.4};
	    \def\spacey{1.4};
	    % x-coordinate
	    \def\xc{0};
	    % y-coordinate
	    \def\yc{0};
	    \def\yu{1*\spacey};
	    \def\yuu{2*\spacey};
	    \def\yuuu{3*\spacey};
	    % coordinate of circle
	    % point at center
	    \coordinate (pc) at (\xc,\yc);
	    % point upper
	    \coordinate (pu) at (\xc,\yu);
	    \coordinate (puu) at (\xc,\yuu);
	    \coordinate (puuu) at (\xc,\yuuu);
	    % horizontal line
	    \draw (\xc-2*\r,\yc) -- (\xc-\r,\yc);
	    \draw (\xc+\r,\yc) -- (\xc+2*\r,\yc);
	    \draw (\xc-2*\r,\yu) -- (\xc-\r,\yu);
	    \draw (\xc+\r,\yu) -- (\xc+2*\r,\yu);
	    \draw (\xc-2*\r,\yuu) -- (\xc-\r,\yuu);
	    \draw (\xc+\r,\yuu) -- (\xc+2*\r,\yuu);
	    \draw (\xc-2*\r,\yuuu) -- (\xc-\r,\yuuu);
	    \draw (\xc+\r,\yuuu) -- (\xc+2*\r,\yuuu);
	    % tensor circle
	    \draw[ten, shift=(puuu)] (-\r,-\r) rectangle (\r,\r);
		\node at (puuu) {\scriptsize $V$};
	    \draw[ten, shift=(puu)] (-\r,-\r) rectangle (\r,\r);
		\node at (puu) {\scriptsize $V^*$};
	    \draw[ten, shift=(pu)] (-\r,-\r) rectangle (\r,\r);
		\node at (pu) {\scriptsize $V$};
	    \draw[ten, shift=(pc)] (-\r,-\r) rectangle (\r,\r);
		\node at (pc) {\scriptsize $V^*$};
	}\right)
	=  \frac{1}{d^2 - 1} \left(
    \diagram{
	    % radius
	    \def\r{0.5};
	    % space
	    \def\spacex{1.4};
	    \def\spacey{1.4};
	    % x-coordinate
	    \def\xc{0};
	    % y-coordinate
	    \def\yc{0};
	    \def\yu{1*\spacey};
	    \def\yuu{2*\spacey};
	    \def\yuuu{3*\spacey};
	    % coordinate of circle
	    % point at center
	    \coordinate (pc) at (\xc,\yc);
	    % point upper
	    \coordinate (pu) at (\xc,\yu);
	    \coordinate (puu) at (\xc,\yuu);
	    \coordinate (puuu) at (\xc,\yuuu);
	    % Bezier curves
	    % \left
        \draw (\xc-2*\r,\yc) .. controls (\xc,\yc) and (\xc,\yu) .. (\xc-2*\r,\yu);
        \draw (\xc-2*\r,\yuu) .. controls (\xc,\yuu) and (\xc,\yuuu) .. (\xc-2*\r,\yuuu);
        % \right
        \draw (\xc+2*\r,\yc) .. controls (\xc,\yc) and (\xc,\yu) .. (\xc+2*\r,\yu);
        \draw (\xc+2*\r,\yuu) .. controls (\xc,\yuu) and (\xc,\yuuu) .. (\xc+2*\r,\yuuu);
    }
    +
    \diagram{
	    % radius
	    \def\r{0.5};
	    % space
	    \def\spacex{1.4};
	    \def\spacey{1.4};
	    % x-coordinate
	    \def\xc{0};
	    % y-coordinate
	    \def\yc{0};
	    \def\yu{1*\spacey};
	    \def\yuu{2*\spacey};
	    \def\yuuu{3*\spacey};
	    % coordinate of circle
	    % point at center
	    \coordinate (pc) at (\xc,\yc);
	    % point upper
	    \coordinate (pu) at (\xc,\yu);
	    \coordinate (puu) at (\xc,\yuu);
	    \coordinate (puuu) at (\xc,\yuuu);
	    % Bezier curves
	    % \left
        \draw (\xc-2*\r,\yc) .. controls (\xc,\yc) and (\xc,\yuuu) .. (\xc-2*\r,\yuuu);
        \draw (\xc-2*\r,\yu) .. controls (\xc-\r,\yu) and (\xc-\r,\yuu) .. (\xc-2*\r,\yuu);
        % \right
        \draw (\xc+2*\r,\yc) .. controls (\xc,\yc) and (\xc,\yuuu) .. (\xc+2*\r,\yuuu);
        \draw (\xc+2*\r,\yu) .. controls (\xc+\r,\yu) and (\xc+\r,\yuu) .. (\xc+2*\r,\yuu);
    }\right)
    - \frac{1}{d(d^2-1)}\left(
        \diagram{
	    % radius
	    \def\r{0.5};
	    % space
	    \def\spacex{1.4};
	    \def\spacey{1.4};
	    % x-coordinate
	    \def\xc{0};
	    % y-coordinate
	    \def\yc{0};
	    \def\yu{1*\spacey};
	    \def\yuu{2*\spacey};
	    \def\yuuu{3*\spacey};
	    % coordinate of circle
	    % point at center
	    \coordinate (pc) at (\xc,\yc);
	    % point upper
	    \coordinate (pu) at (\xc,\yu);
	    \coordinate (puu) at (\xc,\yuu);
	    \coordinate (puuu) at (\xc,\yuuu);
	    % Bezier curves
	    % left
        \draw (\xc-2*\r,\yc) .. controls (\xc,\yc) and (\xc,\yu) .. (\xc-2*\r,\yu);
        \draw (\xc-2*\r,\yuu) .. controls (\xc,\yuu) and (\xc,\yuuu) .. (\xc-2*\r,\yuuu);
        % right
        \draw (\xc+2*\r,\yc) .. controls (\xc,\yc) and (\xc,\yuuu) .. (\xc+2*\r,\yuuu);
        \draw (\xc+2*\r,\yu) .. controls (\xc+\r,\yu) and (\xc+\r,\yuu) .. (\xc+2*\r,\yuu);
    }
    +
    \diagram{
	    % radius
	    \def\r{0.5};
	    % space
	    \def\spacex{1.4};
	    \def\spacey{1.4};
	    % x-coordinate
	    \def\xc{0};
	    % y-coordinate
	    \def\yc{0};
	    \def\yu{1*\spacey};
	    \def\yuu{2*\spacey};
	    \def\yuuu{3*\spacey};
	    % coordinate of circle
	    % point at center
	    \coordinate (pc) at (\xc,\yc);
	    % point upper
	    \coordinate (pu) at (\xc,\yu);
	    \coordinate (puu) at (\xc,\yuu);
	    \coordinate (puuu) at (\xc,\yuuu);
	    % Bezier curves
 	    % \left
        \draw (\xc-2*\r,\yc) .. controls (\xc,\yc) and (\xc,\yuuu) .. (\xc-2*\r,\yuuu);
        \draw (\xc-2*\r,\yu) .. controls (\xc-\r,\yu) and (\xc-\r,\yuu) .. (\xc-2*\r,\yuu);
        % \right
        \draw (\xc+2*\r,\yc) .. controls (\xc,\yc) and (\xc,\yu) .. (\xc+2*\r,\yu);
        \draw (\xc+2*\r,\yuu) .. controls (\xc,\yuu) and (\xc,\yuuu) .. (\xc+2*\r,\yuuu);
    }\right).
\end{equation}
Now we utilize (\ref{Eq:2-moment_element-wise_tn}) to derive a central identity used in the proof in the next section as Lemma \ref{lemma:EV_QVPV}.

\begin{lemma}\label{lemma:EV_QVPV}
Suppose $V\in\mathbb{V}$ is a unitary on the Hilbert space $\mathcal{H}_A\otimes \mathcal{H}_B$ with $\dim{(\mathcal{H}_A)}=d_A$ and $\dim{(\mathcal{H}_B)}=d_B$ where $\mathbb{V}$ is a unitary $2$-design. For any linear operators $P,Q$ on $\mathcal{H}_A\otimes \mathcal{H}_B$, the following identity holds
\begin{equation}\label{Eq:EV_QVPV}
    \begin{aligned}
        \mathbb{E}_{V} \left[ \left\| \tr_B \left(Q V P V^\dagger\right) \right\|^2_2 \right]
        = \frac{1}{d^2-1}\left[
        \|\tr_B Q\|^2_2 \left(|\tr P |^2 - \frac{\|P\|^2_2)}{d}\right)
        + d_A \|Q\|^2_2 \left(\|P\|^2_2 - \frac{ |\tr P|^2}{d} \right)
        \right],
    \end{aligned}
\end{equation}
where $\|\cdot\|_2$ is the Schatten $2$-norm and $d=d_Ad_B$ denotes the dimension of the whole Hilbert space $\mathcal{H}_A\otimes \mathcal{H}_B$.
\end{lemma}
\begin{proof}
Note that $\mathbb{V}$ is a unitary $2$-design and $\left\| \tr_B \left(Q V P V^\dagger\right) \right\|^2_2$ is a polynomial of degree at most $2$ in the entries of $V$. By the definition of unitary $2$-designs in (\ref{Eq:definiton_design}), the left hand side of (\ref{Eq:EV_QVPV}) could be rewritten as
\begin{equation}\label{Eq:EV_QVPV_rewrite}
    \begin{aligned}
        \mathbb{E}_{V} \left[ \left\| \tr_B \left(Q V P V^\dagger\right) \right\|^2_2 \right]
        = \int_{\mathcal{U}(d)} d \mu(V) \tr 
        \left( \tr_B (Q V P V^\dagger) \tr_B (V P^\dagger V^\dagger Q^\dagger) \right).
    \end{aligned}
\end{equation}
Since the Hilbert space $\mathcal{H}_A\otimes\mathcal{H}_B$ has a bipartite tensor product structure, the linear operators on $\mathcal{H}_A\otimes\mathcal{H}_B$ could be represented as $4$-degree tensors. We take the convention for the arrangement of the indices of $V$ and $V^*$ corresponding to $\mathcal{H}_A$, $\mathcal{H}_B$ as follows
\begin{equation}
    \left(
    \diagram{
	    % radius
	    \def\r{0.9};
	    % space
	    \def\spacex{1.4};
	    \def\spacey{1.4};
	    % x-coordinate
	    \def\xc{0};
	    % y-coordinate
	    \def\yc{0};
	    \def\yu{0.5*\r};
	    \def\yd{-0.5*\r};
	    % coordinate of circle
	    \coordinate (pc) at (\xc,\yc);
	    % horizontal line
	    \draw (\xc-2*\r,\yu) -- (\xc-\r,\yu);
	    \draw (\xc+\r,\yu) -- (\xc+2*\r,\yu);
	    \draw (\xc-2*\r,\yd) -- (\xc-\r,\yd);
	    \draw (\xc+\r,\yd) -- (\xc+2*\r,\yd);
	    % tensor circle
	    \draw[ten, shift=(pc)] (-\r,-\r) rectangle (\r,\r);
		\node at (pc) {\normalsize $V$};
		\node at (\xc-3*\r,\yu) {\footnotesize $\mathcal{H}_A$};
		\node at (\xc+3*\r,\yu) {\footnotesize $\mathcal{H}_A$};
		\node at (\xc-3*\r,\yd) {\footnotesize $\mathcal{H}_B$};
		\node at (\xc+3*\r,\yd) {\footnotesize $\mathcal{H}_B$};
	}\right),~~\left(
    \diagram{
	    % radius
	    \def\r{0.9};
	    % space
	    \def\spacex{1.4};
	    \def\spacey{1.4};
	    % x-coordinate
	    \def\xc{0};
	    % y-coordinate
	    \def\yc{0};
	    \def\yu{0.5*\r};
	    \def\yd{-0.5*\r};
	    % coordinate of circle
	    \coordinate (pc) at (\xc,\yc);
	    % horizontal line
	    \draw (\xc-2*\r,\yu) -- (\xc-\r,\yu);
	    \draw (\xc+\r,\yu) -- (\xc+2*\r,\yu);
	    \draw (\xc-2*\r,\yd) -- (\xc-\r,\yd);
	    \draw (\xc+\r,\yd) -- (\xc+2*\r,\yd);
	    % tensor circle
	    \draw[ten, shift=(pc)] (-\r,-\r) rectangle (\r,\r);
		\node at (pc) {\normalsize $V^*$};
		\node at (\xc-3*\r,\yu) {\footnotesize $\mathcal{H}_B$};
		\node at (\xc+3*\r,\yu) {\footnotesize $\mathcal{H}_B$};
		\node at (\xc-3*\r,\yd) {\footnotesize $\mathcal{H}_A$};
		\node at (\xc+3*\r,\yd) {\footnotesize $\mathcal{H}_A$};
	}\right).
\end{equation}
The arrangements of indices for $P,Q$ and $P^*,Q^*$ are the same as $V$ and $V^*$, respectively. The integrand on the right hand side of (\ref{Eq:EV_QVPV_rewrite}) could be represented diagrammatically as
\begin{equation}\label{Eq:OAOBVPV_tn}
    \diagram{
	    % radius
	    \def\r{0.5};
	    % space
	    \def\spacex{1.4};
	    \def\spacey{1.4};
	    % x-coordinate
	    \def\xc{0};
	    \def\xu{\xc+2*\r};
	    \def\xp{\xc+5*\r};
	    \def\xq{\xc-\r};
	    % y-coordinate
	    \def\yc{0};
	    \def\yu{1*\spacey};
	    \def\yuu{2*\spacey};
	    \def\yuuu{3*\spacey};
	    % coordinate of circle
	    % point at center
	    \coordinate (pc) at (\xu,\yc);
	    % point upper
	    \coordinate (pu) at (\xu,\yu);
	    \coordinate (puu) at (\xu,\yuu);
	    \coordinate (puuu) at (\xu,\yuuu);
	 
	    \coordinate (ppu) at (\xp,\yuuu);
	    \coordinate (ppd) at (\xp,\yc);
	    
	    \coordinate (pqu) at (\xq, \yuuu);
	    \coordinate (pqd) at (\xq, \yc);
	    
	    % control Bezier curves
		\draw (\xq-\r,\yuuu+0.5*\r) .. controls (\xq-3*\r,\yuuu+0.5*\r) and (\xq-3*\r,\yc-0.5*\r) .. (\xq-\r,\yc-0.5*\r);
		\draw (\xq-\r,\yuuu-0.5*\r) .. controls (\xq-1.8*\r,\yuuu-0.5*\r) and (\xq-1.8*\r,\yuu+0.5*\r) .. (\xq-\r,\yuu+0.5*\r);
		\draw (\xq-\r,\yuu-0.5*\r) .. controls (\xq-1.8*\r,\yuu-0.5*\r) and (\xq-1.8*\r,\yu+0.5*\r) .. (\xq-\r,\yu+0.5*\r);
		\draw (\xq-\r,\yu-0.5*\r) .. controls (\xq-1.8*\r,\yu-0.5*\r) and (\xq-1.8*\r,\yc+0.5*\r) .. (\xq-\r,\yc+0.5*\r);
		
		\draw (\xp+\r,\yuuu+0.5*\r) .. controls (\xp+3*\r,\yuuu+0.5*\r) and (\xp+3*\r,\yuu-0.5*\r) .. (\xp+\r,\yuu-0.5*\r);
		\draw (\xp+\r,\yuuu-0.5*\r) .. controls (\xp+1.8*\r,\yuuu-0.5*\r) and (\xp+1.8*\r,\yuu+0.5*\r) .. (\xp+\r,\yuu+0.5*\r);
		\draw (\xp+\r,\yu+0.5*\r) .. controls (\xp+3*\r,\yu+0.5*\r) and (\xp+3*\r,\yc-0.5*\r) .. (\xp+\r,\yc-0.5*\r);
		\draw (\xp+\r,\yu-0.5*\r) .. controls (\xp+1.8*\r,\yu-0.5*\r) and (\xp+1.8*\r,\yc+0.5*\r) .. (\xp+\r,\yc+0.5*\r);

	    % horizontal line
	   	\draw (\xc-2*\r,\yuuu+0.5*\r) -- (\xp+\r,\yuuu+0.5*\r);
	    \draw (\xc-2*\r,\yuuu-0.5*\r) -- (\xp+\r,\yuuu-0.5*\r);
	    \draw (\xc-2*\r,\yuu+0.5*\r) -- (\xp+\r,\yuu+0.5*\r);
	    \draw (\xc-2*\r,\yuu-0.5*\r) -- (\xp+\r,\yuu-0.5*\r);
	    \draw (\xc-2*\r,\yu+0.5*\r) -- (\xp+\r,\yu+0.5*\r);
	    \draw (\xc-2*\r,\yu-0.5*\r) -- (\xp+\r,\yu-0.5*\r);
	    \draw (\xc-2*\r,\yc+0.5*\r) -- (\xp+\r,\yc+0.5*\r);
	    \draw (\xc-2*\r,\yc-0.5*\r) -- (\xp+\r,\yc-0.5*\r);
	    
 	    % tensor square
 	    \draw[ten, shift=(puuu)] (-\r,-\r) rectangle (\r,\r);
 	    \node at (puuu) {\scriptsize $V$};
 	    \draw[ten, shift=(puu)] (-\r,-\r) rectangle (\r,\r);
 		\node at (puu) {\scriptsize $V^*$};
 	    \draw[ten, shift=(pu)] (-\r,-\r) rectangle (\r,\r);
 		\node at (pu) {\scriptsize $V$};
 	    \draw[ten, shift=(pc)] (-\r,-\r) rectangle (\r,\r);
 		\node at (pc) {\scriptsize $V^*$};
 		
 		\draw[ten, shift=(ppu)] (-\r,-\r) rectangle (\r,\r);
 	    \node at (ppu) {\scriptsize $P$};
 	    \draw[ten, shift=(ppd)] (-\r,-\r) rectangle (\r,\r);
 	    \node at (ppd) {\scriptsize $P^*$};

  		\draw[ten, shift=(pqu)] (-\r,-\r) rectangle (\r,\r);
  		\node at (pqu) {\scriptsize $Q$};
  		 \draw[ten, shift=(pqd)] (-\r,-\r) rectangle (\r,\r);
  		\node at (pqd) {\scriptsize $Q^*$};
	},
\end{equation}
Combining (\ref{Eq:2-moment_element-wise_tn}), (\ref{Eq:EV_QVPV_rewrite}) and (\ref{Eq:OAOBVPV_tn}), the left hand side of (\ref{Eq:EV_QVPV}) is equal to
\begin{equation}
\begin{aligned}
    & \mathbb{E}_{V} \left\| \tr_B \left(Q V P V^\dagger\right) \right\|^2_2 \\
    =& \frac{1}{d^2-1} \left(
    \diagram{
    	% radius
	    \def\r{0.5};
	    % space
	    \def\spacex{1.4};
	    \def\spacey{1.4};
	    % x-coordinate
	    \def\xc{0};
	    \def\xu{\xc+2*\r};
	    \def\xp{\xc+5*\r};
	    \def\xq{\xc-\r};
	    % y-coordinate
	    \def\yc{0};
	    \def\yu{1*\spacey};
	    \def\yuu{2*\spacey};
	    \def\yuuu{3*\spacey};
	    % coordinate of circle
	    % point at center
	    \coordinate (pc) at (\xu,\yc);
	    % point upper
	    \coordinate (pu) at (\xu,\yu);
	    \coordinate (puu) at (\xu,\yuu);
	    \coordinate (puuu) at (\xu,\yuuu);
	 
	    \coordinate (ppu) at (\xp,\yuuu);
	    \coordinate (ppd) at (\xp,\yc);
	    
	    \coordinate (pqu) at (\xq, \yuuu);
	    \coordinate (pqd) at (\xq, \yc);
	    
	    % control Bezier curves
	    % left side of Q
		\draw (\xq-\r,\yuuu+0.5*\r) .. controls (\xq-3*\r,\yuuu+0.5*\r) and (\xq-3*\r,\yc-0.5*\r) .. (\xq-\r,\yc-0.5*\r);
		\draw (\xq-\r,\yuuu-0.5*\r) .. controls (\xq-1.8*\r,\yuuu-0.5*\r) and (\xq-1.8*\r,\yuu+0.5*\r) .. (\xq-\r,\yuu+0.5*\r);
		\draw (\xq-\r,\yuu-0.5*\r) .. controls (\xq-1.8*\r,\yuu-0.5*\r) and (\xq-1.8*\r,\yu+0.5*\r) .. (\xq-\r,\yu+0.5*\r);
		\draw (\xq-\r,\yu-0.5*\r) .. controls (\xq-1.8*\r,\yu-0.5*\r) and (\xq-1.8*\r,\yc+0.5*\r) .. (\xq-\r,\yc+0.5*\r);
		% right side of P
		\draw (\xp+\r,\yuuu+0.5*\r) .. controls (\xp+3*\r,\yuuu+0.5*\r) and (\xp+3*\r,\yuu-0.5*\r) .. (\xp+\r,\yuu-0.5*\r);
		\draw (\xp+\r,\yuuu-0.5*\r) .. controls (\xp+1.8*\r,\yuuu-0.5*\r) and (\xp+1.8*\r,\yuu+0.5*\r) .. (\xp+\r,\yuu+0.5*\r);
		\draw (\xp+\r,\yu+0.5*\r) .. controls (\xp+3*\r,\yu+0.5*\r) and (\xp+3*\r,\yc-0.5*\r) .. (\xp+\r,\yc-0.5*\r);
		\draw (\xp+\r,\yu-0.5*\r) .. controls (\xp+1.8*\r,\yu-0.5*\r) and (\xp+1.8*\r,\yc+0.5*\r) .. (\xp+\r,\yc+0.5*\r);
		% right side of Q
		\draw (\xq+\r,\yuuu+0.5*\r) .. controls (\xq+3*\r,\yuuu+0.5*\r) and (\xq+3*\r,\yuu-0.5*\r) .. (\xq+\r,\yuu-0.5*\r);
		\draw (\xq+\r,\yuuu-0.5*\r) .. controls (\xq+1.8*\r,\yuuu-0.5*\r) and (\xq+1.8*\r,\yuu+0.5*\r) .. (\xq+\r,\yuu+0.5*\r);
		\draw (\xq+\r,\yu+0.5*\r) .. controls (\xq+3*\r,\yu+0.5*\r) and (\xq+3*\r,\yc-0.5*\r) .. (\xq+\r,\yc-0.5*\r);
		\draw (\xq+\r,\yu-0.5*\r) .. controls (\xq+1.8*\r,\yu-0.5*\r) and (\xq+1.8*\r,\yc+0.5*\r) .. (\xq+\r,\yc+0.5*\r);
		% left side of P
		\draw (\xp-\r,\yuuu+0.5*\r) .. controls (\xp-3*\r,\yuuu+0.5*\r) and (\xp-3*\r,\yuu-0.5*\r) .. (\xp-\r,\yuu-0.5*\r);
		\draw (\xp-\r,\yuuu-0.5*\r) .. controls (\xp-1.8*\r,\yuuu-0.5*\r) and (\xp-1.8*\r,\yuu+0.5*\r) .. (\xp-\r,\yuu+0.5*\r);
		\draw (\xp-\r,\yu+0.5*\r) .. controls (\xp-3*\r,\yu+0.5*\r) and (\xp-3*\r,\yc-0.5*\r) .. (\xp-\r,\yc-0.5*\r);
		\draw (\xp-\r,\yu-0.5*\r) .. controls (\xp-1.8*\r,\yu-0.5*\r) and (\xp-1.8*\r,\yc+0.5*\r) .. (\xp-\r,\yc+0.5*\r);
		
	    % horizontal line Q
	    \draw (\xq-\r,\yuu+0.5*\r) -- (\xq+\r,\yuu+0.5*\r);
	    \draw (\xq-\r,\yuu-0.5*\r) -- (\xq+\r,\yuu-0.5*\r);
	    \draw (\xq-\r,\yu+0.5*\r) -- (\xq+\r,\yu+0.5*\r);
	    \draw (\xq-\r,\yu-0.5*\r) -- (\xq+\r,\yu-0.5*\r);
	    
	    % horizontal line P
	    \draw (\xp-\r,\yuu+0.5*\r) -- (\xp+\r,\yuu+0.5*\r);
	    \draw (\xp-\r,\yuu-0.5*\r) -- (\xp+\r,\yuu-0.5*\r);
	    \draw (\xp-\r,\yu+0.5*\r) -- (\xp+\r,\yu+0.5*\r);
	    \draw (\xp-\r,\yu-0.5*\r) -- (\xp+\r,\yu-0.5*\r);
 		
 		\draw[ten, shift=(ppu)] (-\r,-\r) rectangle (\r,\r);
 	    \node at (ppu) {\scriptsize $P$};
 	    \draw[ten, shift=(ppd)] (-\r,-\r) rectangle (\r,\r);
 	    \node at (ppd) {\scriptsize $P^*$};

  		\draw[ten, shift=(pqu)] (-\r,-\r) rectangle (\r,\r);
  		\node at (pqu) {\scriptsize $Q$};
  		 \draw[ten, shift=(pqd)] (-\r,-\r) rectangle (\r,\r);
  		\node at (pqd) {\scriptsize $Q^*$};
	}
	+
	\diagram{
    	% radius
	    \def\r{0.5};
	    % space
	    \def\spacex{1.4};
	    \def\spacey{1.4};
	    % x-coordinate
	    \def\xc{0};
	    \def\xu{\xc+2*\r};
	    \def\xp{\xc+5*\r};
	    \def\xq{\xc-\r};
	    % y-coordinate
	    \def\yc{0};
	    \def\yu{1*\spacey};
	    \def\yuu{2*\spacey};
	    \def\yuuu{3*\spacey};
	    % coordinate of circle
	    % point at center
	    \coordinate (pc) at (\xu,\yc);
	    % point upper
	    \coordinate (pu) at (\xu,\yu);
	    \coordinate (puu) at (\xu,\yuu);
	    \coordinate (puuu) at (\xu,\yuuu);
	 
	    \coordinate (ppu) at (\xp,\yuuu);
	    \coordinate (ppd) at (\xp,\yc);
	    
	    \coordinate (pqu) at (\xq, \yuuu);
	    \coordinate (pqd) at (\xq, \yc);
	    
	    % control Bezier curves
	    % left side of Q
		\draw (\xq-\r,\yuuu+0.5*\r) .. controls (\xq-3*\r,\yuuu+0.5*\r) and (\xq-3*\r,\yc-0.5*\r) .. (\xq-\r,\yc-0.5*\r);
		\draw (\xq-\r,\yuuu-0.5*\r) .. controls (\xq-1.8*\r,\yuuu-0.5*\r) and (\xq-1.8*\r,\yuu+0.5*\r) .. (\xq-\r,\yuu+0.5*\r);
		\draw (\xq-\r,\yuu-0.5*\r) .. controls (\xq-1.8*\r,\yuu-0.5*\r) and (\xq-1.8*\r,\yu+0.5*\r) .. (\xq-\r,\yu+0.5*\r);
		\draw (\xq-\r,\yu-0.5*\r) .. controls (\xq-1.8*\r,\yu-0.5*\r) and (\xq-1.8*\r,\yc+0.5*\r) .. (\xq-\r,\yc+0.5*\r);
		% right side of P
		\draw (\xp+\r,\yuuu+0.5*\r) .. controls (\xp+3*\r,\yuuu+0.5*\r) and (\xp+3*\r,\yuu-0.5*\r) .. (\xp+\r,\yuu-0.5*\r);
		\draw (\xp+\r,\yuuu-0.5*\r) .. controls (\xp+1.8*\r,\yuuu-0.5*\r) and (\xp+1.8*\r,\yuu+0.5*\r) .. (\xp+\r,\yuu+0.5*\r);
		\draw (\xp+\r,\yu+0.5*\r) .. controls (\xp+3*\r,\yu+0.5*\r) and (\xp+3*\r,\yc-0.5*\r) .. (\xp+\r,\yc-0.5*\r);
		\draw (\xp+\r,\yu-0.5*\r) .. controls (\xp+1.8*\r,\yu-0.5*\r) and (\xp+1.8*\r,\yc+0.5*\r) .. (\xp+\r,\yc+0.5*\r);
		% right side of Q
		\draw (\xq+\r,\yuuu+0.5*\r) .. controls (\xq+3*\r,\yuuu+0.5*\r) and (\xq+3*\r,\yc-0.5*\r) .. (\xq+\r,\yc-0.5*\r);
		\draw (\xq+\r,\yuuu-0.5*\r) .. controls (\xq+2.6*\r,\yuuu-0.5*\r) and (\xq+2.6*\r,\yc+0.5*\r) .. (\xq+\r,\yc+0.5*\r);
		\draw (\xq+\r,\yu+0.5*\r) .. controls (\xq+1.8*\r,\yu+0.5*\r) and (\xq+1.8*\r,\yuu-0.5*\r) .. (\xq+\r,\yuu-0.5*\r);
		\draw (\xq+\r,\yu-0.5*\r) .. controls (\xq+2.2*\r,\yu-0.5*\r) and (\xq+2.2*\r,\yuu+0.5*\r) .. (\xq+\r,\yuu+0.5*\r);
		% left side of P
		\draw (\xp-\r,\yuuu+0.5*\r) .. controls (\xp-3*\r,\yuuu+0.5*\r) and (\xp-3*\r,\yc-0.5*\r) .. (\xp-\r,\yc-0.5*\r);
		\draw (\xp-\r,\yuuu-0.5*\r) .. controls (\xp-2.6*\r,\yuuu-0.5*\r) and (\xp-2.6*\r,\yc+0.5*\r) .. (\xp-\r,\yc+0.5*\r);
		\draw (\xp-\r,\yu+0.5*\r) .. controls (\xp-1.8*\r,\yu+0.5*\r) and (\xp-1.8*\r,\yuu-0.5*\r) .. (\xp-\r,\yuu-0.5*\r);
		\draw (\xp-\r,\yu-0.5*\r) .. controls (\xp-2.2*\r,\yu-0.5*\r) and (\xp-2.2*\r,\yuu+0.5*\r) .. (\xp-\r,\yuu+0.5*\r);
		
	    % horizontal line Q
	    \draw (\xq-\r,\yuu+0.5*\r) -- (\xq+\r,\yuu+0.5*\r);
	    \draw (\xq-\r,\yuu-0.5*\r) -- (\xq+\r,\yuu-0.5*\r);
	    \draw (\xq-\r,\yu+0.5*\r) -- (\xq+\r,\yu+0.5*\r);
	    \draw (\xq-\r,\yu-0.5*\r) -- (\xq+\r,\yu-0.5*\r);
	    
	    % horizontal line P
	    \draw (\xp-\r,\yuu+0.5*\r) -- (\xp+\r,\yuu+0.5*\r);
	    \draw (\xp-\r,\yuu-0.5*\r) -- (\xp+\r,\yuu-0.5*\r);
	    \draw (\xp-\r,\yu+0.5*\r) -- (\xp+\r,\yu+0.5*\r);
	    \draw (\xp-\r,\yu-0.5*\r) -- (\xp+\r,\yu-0.5*\r);
 		
 		\draw[ten, shift=(ppu)] (-\r,-\r) rectangle (\r,\r);
 	    \node at (ppu) {\scriptsize $P$};
 	    \draw[ten, shift=(ppd)] (-\r,-\r) rectangle (\r,\r);
 	    \node at (ppd) {\scriptsize $P^*$};

  		\draw[ten, shift=(pqu)] (-\r,-\r) rectangle (\r,\r);
  		\node at (pqu) {\scriptsize $Q$};
  		 \draw[ten, shift=(pqd)] (-\r,-\r) rectangle (\r,\r);
  		\node at (pqd) {\scriptsize $Q^*$};
	}
	\right) \\
	& - \frac{1}{d(d^2-1)} \left(
    \diagram{
    	% radius
	    \def\r{0.5};
	    % space
	    \def\spacex{1.4};
	    \def\spacey{1.4};
	    % x-coordinate
	    \def\xc{0};
	    \def\xu{\xc+2*\r};
	    \def\xp{\xc+5*\r};
	    \def\xq{\xc-\r};
	    % y-coordinate
	    \def\yc{0};
	    \def\yu{1*\spacey};
	    \def\yuu{2*\spacey};
	    \def\yuuu{3*\spacey};
	    % coordinate of circle
	    % point at center
	    \coordinate (pc) at (\xu,\yc);
	    % point upper
	    \coordinate (pu) at (\xu,\yu);
	    \coordinate (puu) at (\xu,\yuu);
	    \coordinate (puuu) at (\xu,\yuuu);
	 
	    \coordinate (ppu) at (\xp,\yuuu);
	    \coordinate (ppd) at (\xp,\yc);
	    
	    \coordinate (pqu) at (\xq, \yuuu);
	    \coordinate (pqd) at (\xq, \yc);
	    
	    % control Bezier curves
	    % left side of Q
		\draw (\xq-\r,\yuuu+0.5*\r) .. controls (\xq-3*\r,\yuuu+0.5*\r) and (\xq-3*\r,\yc-0.5*\r) .. (\xq-\r,\yc-0.5*\r);
		\draw (\xq-\r,\yuuu-0.5*\r) .. controls (\xq-1.8*\r,\yuuu-0.5*\r) and (\xq-1.8*\r,\yuu+0.5*\r) .. (\xq-\r,\yuu+0.5*\r);
		\draw (\xq-\r,\yuu-0.5*\r) .. controls (\xq-1.8*\r,\yuu-0.5*\r) and (\xq-1.8*\r,\yu+0.5*\r) .. (\xq-\r,\yu+0.5*\r);
		\draw (\xq-\r,\yu-0.5*\r) .. controls (\xq-1.8*\r,\yu-0.5*\r) and (\xq-1.8*\r,\yc+0.5*\r) .. (\xq-\r,\yc+0.5*\r);
		% right side of P
		\draw (\xp+\r,\yuuu+0.5*\r) .. controls (\xp+3*\r,\yuuu+0.5*\r) and (\xp+3*\r,\yuu-0.5*\r) .. (\xp+\r,\yuu-0.5*\r);
		\draw (\xp+\r,\yuuu-0.5*\r) .. controls (\xp+1.8*\r,\yuuu-0.5*\r) and (\xp+1.8*\r,\yuu+0.5*\r) .. (\xp+\r,\yuu+0.5*\r);
		\draw (\xp+\r,\yu+0.5*\r) .. controls (\xp+3*\r,\yu+0.5*\r) and (\xp+3*\r,\yc-0.5*\r) .. (\xp+\r,\yc-0.5*\r);
		\draw (\xp+\r,\yu-0.5*\r) .. controls (\xp+1.8*\r,\yu-0.5*\r) and (\xp+1.8*\r,\yc+0.5*\r) .. (\xp+\r,\yc+0.5*\r);
		% right side of Q
		\draw (\xq+\r,\yuuu+0.5*\r) .. controls (\xq+3*\r,\yuuu+0.5*\r) and (\xq+3*\r,\yuu-0.5*\r) .. (\xq+\r,\yuu-0.5*\r);
		\draw (\xq+\r,\yuuu-0.5*\r) .. controls (\xq+1.8*\r,\yuuu-0.5*\r) and (\xq+1.8*\r,\yuu+0.5*\r) .. (\xq+\r,\yuu+0.5*\r);
		\draw (\xq+\r,\yu+0.5*\r) .. controls (\xq+3*\r,\yu+0.5*\r) and (\xq+3*\r,\yc-0.5*\r) .. (\xq+\r,\yc-0.5*\r);
		\draw (\xq+\r,\yu-0.5*\r) .. controls (\xq+1.8*\r,\yu-0.5*\r) and (\xq+1.8*\r,\yc+0.5*\r) .. (\xq+\r,\yc+0.5*\r);
		% left side of P
		\draw (\xp-\r,\yuuu+0.5*\r) .. controls (\xp-3*\r,\yuuu+0.5*\r) and (\xp-3*\r,\yc-0.5*\r) .. (\xp-\r,\yc-0.5*\r);
		\draw (\xp-\r,\yuuu-0.5*\r) .. controls (\xp-2.6*\r,\yuuu-0.5*\r) and (\xp-2.6*\r,\yc+0.5*\r) .. (\xp-\r,\yc+0.5*\r);
		\draw (\xp-\r,\yu+0.5*\r) .. controls (\xp-1.8*\r,\yu+0.5*\r) and (\xp-1.8*\r,\yuu-0.5*\r) .. (\xp-\r,\yuu-0.5*\r);
		\draw (\xp-\r,\yu-0.5*\r) .. controls (\xp-2.2*\r,\yu-0.5*\r) and (\xp-2.2*\r,\yuu+0.5*\r) .. (\xp-\r,\yuu+0.5*\r);
		
	    % horizontal line Q
	    \draw (\xq-\r,\yuu+0.5*\r) -- (\xq+\r,\yuu+0.5*\r);
	    \draw (\xq-\r,\yuu-0.5*\r) -- (\xq+\r,\yuu-0.5*\r);
	    \draw (\xq-\r,\yu+0.5*\r) -- (\xq+\r,\yu+0.5*\r);
	    \draw (\xq-\r,\yu-0.5*\r) -- (\xq+\r,\yu-0.5*\r);
	    
	    % horizontal line P
	    \draw (\xp-\r,\yuu+0.5*\r) -- (\xp+\r,\yuu+0.5*\r);
	    \draw (\xp-\r,\yuu-0.5*\r) -- (\xp+\r,\yuu-0.5*\r);
	    \draw (\xp-\r,\yu+0.5*\r) -- (\xp+\r,\yu+0.5*\r);
	    \draw (\xp-\r,\yu-0.5*\r) -- (\xp+\r,\yu-0.5*\r);
 		
 		\draw[ten, shift=(ppu)] (-\r,-\r) rectangle (\r,\r);
 	    \node at (ppu) {\scriptsize $P$};
 	    \draw[ten, shift=(ppd)] (-\r,-\r) rectangle (\r,\r);
 	    \node at (ppd) {\scriptsize $P^*$};

  		\draw[ten, shift=(pqu)] (-\r,-\r) rectangle (\r,\r);
  		\node at (pqu) {\scriptsize $Q$};
  		 \draw[ten, shift=(pqd)] (-\r,-\r) rectangle (\r,\r);
  		\node at (pqd) {\scriptsize $Q^*$};
	}
	+
	\diagram{
    	% radius
	    \def\r{0.5};
	    % space
	    \def\spacex{1.4};
	    \def\spacey{1.4};
	    % x-coordinate
	    \def\xc{0};
	    \def\xu{\xc+2*\r};
	    \def\xp{\xc+5*\r};
	    \def\xq{\xc-\r};
	    % y-coordinate
	    \def\yc{0};
	    \def\yu{1*\spacey};
	    \def\yuu{2*\spacey};
	    \def\yuuu{3*\spacey};
	    % coordinate of circle
	    % point at center
	    \coordinate (pc) at (\xu,\yc);
	    % point upper
	    \coordinate (pu) at (\xu,\yu);
	    \coordinate (puu) at (\xu,\yuu);
	    \coordinate (puuu) at (\xu,\yuuu);
	 
	    \coordinate (ppu) at (\xp,\yuuu);
	    \coordinate (ppd) at (\xp,\yc);
	    
	    \coordinate (pqu) at (\xq, \yuuu);
	    \coordinate (pqd) at (\xq, \yc);
	    
	    % control Bezier curves
	    % left side of Q
		\draw (\xq-\r,\yuuu+0.5*\r) .. controls (\xq-3*\r,\yuuu+0.5*\r) and (\xq-3*\r,\yc-0.5*\r) .. (\xq-\r,\yc-0.5*\r);
		\draw (\xq-\r,\yuuu-0.5*\r) .. controls (\xq-1.8*\r,\yuuu-0.5*\r) and (\xq-1.8*\r,\yuu+0.5*\r) .. (\xq-\r,\yuu+0.5*\r);
		\draw (\xq-\r,\yuu-0.5*\r) .. controls (\xq-1.8*\r,\yuu-0.5*\r) and (\xq-1.8*\r,\yu+0.5*\r) .. (\xq-\r,\yu+0.5*\r);
		\draw (\xq-\r,\yu-0.5*\r) .. controls (\xq-1.8*\r,\yu-0.5*\r) and (\xq-1.8*\r,\yc+0.5*\r) .. (\xq-\r,\yc+0.5*\r);
		% right side of P
		\draw (\xp+\r,\yuuu+0.5*\r) .. controls (\xp+3*\r,\yuuu+0.5*\r) and (\xp+3*\r,\yuu-0.5*\r) .. (\xp+\r,\yuu-0.5*\r);
		\draw (\xp+\r,\yuuu-0.5*\r) .. controls (\xp+1.8*\r,\yuuu-0.5*\r) and (\xp+1.8*\r,\yuu+0.5*\r) .. (\xp+\r,\yuu+0.5*\r);
		\draw (\xp+\r,\yu+0.5*\r) .. controls (\xp+3*\r,\yu+0.5*\r) and (\xp+3*\r,\yc-0.5*\r) .. (\xp+\r,\yc-0.5*\r);
		\draw (\xp+\r,\yu-0.5*\r) .. controls (\xp+1.8*\r,\yu-0.5*\r) and (\xp+1.8*\r,\yc+0.5*\r) .. (\xp+\r,\yc+0.5*\r);
		% right side of Q
		\draw (\xq+\r,\yuuu+0.5*\r) .. controls (\xq+3*\r,\yuuu+0.5*\r) and (\xq+3*\r,\yc-0.5*\r) .. (\xq+\r,\yc-0.5*\r);
		\draw (\xq+\r,\yuuu-0.5*\r) .. controls (\xq+2.6*\r,\yuuu-0.5*\r) and (\xq+2.6*\r,\yc+0.5*\r) .. (\xq+\r,\yc+0.5*\r);
		\draw (\xq+\r,\yu+0.5*\r) .. controls (\xq+1.8*\r,\yu+0.5*\r) and (\xq+1.8*\r,\yuu-0.5*\r) .. (\xq+\r,\yuu-0.5*\r);
		\draw (\xq+\r,\yu-0.5*\r) .. controls (\xq+2.2*\r,\yu-0.5*\r) and (\xq+2.2*\r,\yuu+0.5*\r) .. (\xq+\r,\yuu+0.5*\r);
		% left side of P
		\draw (\xp-\r,\yuuu+0.5*\r) .. controls (\xp-3*\r,\yuuu+0.5*\r) and (\xp-3*\r,\yuu-0.5*\r) .. (\xp-\r,\yuu-0.5*\r);
		\draw (\xp-\r,\yuuu-0.5*\r) .. controls (\xp-1.8*\r,\yuuu-0.5*\r) and (\xp-1.8*\r,\yuu+0.5*\r) .. (\xp-\r,\yuu+0.5*\r);
		\draw (\xp-\r,\yu+0.5*\r) .. controls (\xp-3*\r,\yu+0.5*\r) and (\xp-3*\r,\yc-0.5*\r) .. (\xp-\r,\yc-0.5*\r);
		\draw (\xp-\r,\yu-0.5*\r) .. controls (\xp-1.8*\r,\yu-0.5*\r) and (\xp-1.8*\r,\yc+0.5*\r) .. (\xp-\r,\yc+0.5*\r);
		
	    % horizontal line Q
	    \draw (\xq-\r,\yuu+0.5*\r) -- (\xq+\r,\yuu+0.5*\r);
	    \draw (\xq-\r,\yuu-0.5*\r) -- (\xq+\r,\yuu-0.5*\r);
	    \draw (\xq-\r,\yu+0.5*\r) -- (\xq+\r,\yu+0.5*\r);
	    \draw (\xq-\r,\yu-0.5*\r) -- (\xq+\r,\yu-0.5*\r);
	    
	    % horizontal line P
	    \draw (\xp-\r,\yuu+0.5*\r) -- (\xp+\r,\yuu+0.5*\r);
	    \draw (\xp-\r,\yuu-0.5*\r) -- (\xp+\r,\yuu-0.5*\r);
	    \draw (\xp-\r,\yu+0.5*\r) -- (\xp+\r,\yu+0.5*\r);
	    \draw (\xp-\r,\yu-0.5*\r) -- (\xp+\r,\yu-0.5*\r);
 		
 		\draw[ten, shift=(ppu)] (-\r,-\r) rectangle (\r,\r);
 	    \node at (ppu) {\scriptsize $P$};
 	    \draw[ten, shift=(ppd)] (-\r,-\r) rectangle (\r,\r);
 	    \node at (ppd) {\scriptsize $P^*$};

  		\draw[ten, shift=(pqu)] (-\r,-\r) rectangle (\r,\r);
  		\node at (pqu) {\scriptsize $Q$};
  		 \draw[ten, shift=(pqd)] (-\r,-\r) rectangle (\r,\r);
  		\node at (pqd) {\scriptsize $Q^*$};
	}
	\right) \\
	=& \frac{1}{d^2-1}\left[ \tr(\tr_B Q^\dagger \tr_B Q) \left(|\tr(P)|^2 - \frac{\tr(P^\dagger P)}{d}\right) + d_A \tr(Q^\dagger Q) \left(\tr(P^\dagger P) - \frac{|\tr P|^2}{d} \right) \right],
\end{aligned}
\end{equation}
which is exactly the desired identity (\ref{Eq:EV_QVPV}).
\end{proof}

Then, we will explicitly write down several special cases of Lemma \ref{lemma:EV_QVPV} for the sake of convenience.

\begin{corollary}\label{corollary:ave_purity_from_mixed}
Suppose $V\in\mathbb{V}$ is a unitary on the Hilbert space $\mathcal{H}_A\otimes \mathcal{H}_B$ with $\dim{(\mathcal{H}_A)}=d_A$ and $\dim{(\mathcal{H}_B)}=d_B$ where $\mathbb{V}$ is a unitary $2$-design. Let $\rho$ be an arbitrary density matrix on $\mathcal{H}_A\otimes \mathcal{H}_B$ and $\rho_A=\tr_B(V \rho V^\dagger)$ be the reduced density matrix on $\mathcal{H}_A$ from $V \rho V^\dagger $. The expectation of the purity of $\rho_A$ is
\begin{equation}\label{Eq:ave_purity_from_mixed}
    \mathbb{E}_{V}\left[\tr\left( \rho_A^2 \right)\right] 
    = \frac{(d_A^2 - 1)d_B}{ d^2 - 1} \tr(\rho^2) + \frac{(d_B^2 - 1)d_A}{d^2 - 1}.
\end{equation}
where $d=d_Ad_B$ denotes the dimension of the whole Hilbert space $\mathcal{H}_A\otimes \mathcal{H}_B$. Since pure states satisfy $\tr(\rho^2)=1$, {\rm (\ref{Eq:ave_purity_from_mixed})} can be further simplified for pure states as
\begin{equation}
    \mathbb{E}_{V}\left[\tr\left( \rho_A^2 \right)\right] = \frac{d_A+d_B}{d_Ad_B+1}.
\end{equation}
\end{corollary}
\begin{proof}
This is a special case of Lemma \ref{lemma:EV_QVPV} by taking $P=\rho$ and $Q=I_A\otimes I_B$.
\end{proof}
\begin{corollary}\label{corollary:EV_IAtracelessOBVrhoV}
Suppose $V\in\mathbb{V}$ is a unitary on the Hilbert space $\mathcal{H}_A\otimes \mathcal{H}_B$ with $\dim{(\mathcal{H}_A)}=d_A$ and $\dim{(\mathcal{H}_B)}=d_B$ where $\mathbb{V}$ is a unitary $2$-design. Let $\rho$ be an arbitrary density matrix on $\mathcal{H}_A\otimes \mathcal{H}_B$. For any traceless operator $O_B$ on $\mathcal{H}_B$, the following identity holds
\begin{equation}\label{Eq:EUIAtracelessOBUrhoU}
    \mathbb{E}_{V} \left[ \left\| \tr_B \left( I_A\otimes O_B V \rho V^\dagger \right) \right \|^2_2 \right]
    = \frac{d_A^2 \|O_B\|^2_2 }{d^2-1} \left( \tr(\rho^2) -\frac{1}{d} \right),
\end{equation}
where $d=d_Ad_B$ denotes the dimension of the whole Hilbert space $\mathcal{H}_A\otimes \mathcal{H}_B$ and $I_A$ is the identity on $\mathcal{H}_A$.
\end{corollary}
\begin{proof}
This is a special case of Lemma \ref{lemma:EV_QVPV} by taking $P=\rho$, $Q=I_A\otimes O_B$ with $\tr(O_B)=0$.
\end{proof}
\begin{corollary}\label{corollary:EV_2-norm_OAOBtracelessP}
Suppose $V\in\mathbb{V}$ is a unitary on the Hilbert space $\mathcal{H}_A\otimes \mathcal{H}_B$ with $\dim{(\mathcal{H}_A)}=d_A$ and $\dim{(\mathcal{H}_B)}=d_B$ where $\mathbb{V}$ is a unitary $2$-design. For any traceless operator $P$ on $\mathcal{H}_A\otimes \mathcal{H}_B$ and any linear operators $O_A,O_B$ on $\mathcal{H}_A$, $\mathcal{H}_B$ respectively, the following identity holds
\begin{equation}\label{Eq:EV_QVPVtracelessP}
    \begin{aligned}
        \mathbb{E}_{V} \left[ \left\| \tr_B \left( O_A\otimes O_B V P V^\dagger\right) \right\|^2_2 \right]
        = \frac{\|O_A\|^2_2 \|P\|^2_2 }{d^2-1}\left[
        d_A \|O_B\|^2_2 - \frac{|\tr O_B|^2}{d}
        \right],
    \end{aligned}
\end{equation}
where $d=d_Ad_B$ denotes the dimension of the whole Hilbert space $\mathcal{H}_A\otimes \mathcal{H}_B$.
\end{corollary}
\begin{proof}
This is a special case of Lemma \ref{lemma:EV_QVPV} by taking $\tr (P)=0$ and $Q=O_A\otimes O_B$.
\end{proof}

In the end of this section, we recall several fundamental inequalities in linear algebra and probability theory to make our proofs in the next section more self-contained.
\begin{lemma}\label{lemma:holder}{\rm (Hölder's inequality for tracial matrices)}
For any linear operators $X,Y$, the following inequality holds
\begin{equation}
    \left| \tr (X^\dagger Y) \right| \leq  \|X\|_p \|Y\|_q,
\end{equation}
where $p,q$ satisfy $\frac{1}{p}+\frac{1}{q} = 1$ and $\|\cdot\|_p$ denotes the Schatten $p$-norm defined by $\|A\|_p = \left(\tr|A|^p\right)^{1/p}$, $|A|=\sqrt{A^\dagger A}$.
\end{lemma}
\begin{lemma}\label{lemma:partial_trace_norm}{\rm (Partial trace monotonicity)}
For any linear operator $H$ on the Hilbert space $\mathcal{H}_A\otimes \mathcal{H}_B$ with $\dim\mathcal{H}_{B}=d_B$, the following inequality holds  {\rm~\cite{Rastegin2012}}
\begin{equation}
    \| \tr_B H \|_{p} \leq d_B^{(p-1)/p} \|H\|_{p}.
\end{equation}
Namely, the Schatten $p$-norm is non-increasing under partial tracing up to a constant coefficient. Specially, we have
\begin{equation}
    \| \tr_B H \|_{1} \leq \|H\|_{1},~~~
    \| \tr_B H \|_{2} \leq \sqrt{d_B} \|H\|_{2},~~~
    \| \tr_B H \|_{\infty} \leq d_B \|H\|_{\infty}.
\end{equation}
\end{lemma}
\begin{lemma}\label{lemma:markov_inequality} {\rm (Markov's inequality)}
Let $X$ be a random variable taking non-negative real value. For any $\epsilon>0$, the following inequality holds
\begin{equation}
    \operatorname{Pr}[X\geq \epsilon] \leq \frac{\mathbb{E}[X]}{\epsilon},
\end{equation}
where $\operatorname{Pr}[X\geq \epsilon]$ denotes the probability of $X\geq \epsilon$ and $\mathbb{E}[X]$ denotes the expectation of the random variable $X$.
\end{lemma}
\begin{lemma}\label{lemma:jensen_inequality} {\rm (Jensen's inequality)}
Let $X$ be a random variable and $f:\mathbb{R}\rightarrow\mathbb{R}$ is a convex function. The following inequality holds
\begin{equation}
    f(\mathbb{E}[X]) \leq \mathbb{E}[f(X)].
\end{equation}
\end{lemma}
\begin{lemma}\label{lemma:var<expec}
Suppose that $X$ is a random variable taking real values in $[0,a]$. The following inequality holds
\begin{equation}
    \operatorname{Var}[X] \leq a \cdot \mathbb{E}[X]. 
\end{equation}
\end{lemma}
\begin{proof}
According to the relation $x^2\leq ax$, we have
\begin{equation}
    \operatorname{Var}[X] \leq \mathbb{E}[X^2] \leq \mathbb{E}[a X] 
    = a \cdot \mathbb{E}[X]. 
\end{equation}
\end{proof}

% \begin{equation}
%     E(X^2) = \frac{\sum x_i^2}{n} \leq \frac{(\sum |x_i|)^2}{n} = n E(|x|)^2
% \end{equation}
% \begin{equation}
%     \operatorname{Var}(X) \leq n E(|x|)^2 - E^2(x) \leq n E(|x|)^2 \in \mathcal{O}(n 2^{-n}) 
% \end{equation}
%%%%%%%%%%%%%%%%%%%%%%%%%%%%%%%%%%%%%%%%%%%%%%%%%%%%%%%%%%%%%%%%%%%

\subsection{Proof of Theorem~\textcolor{blue}{1}}\label{supplementary:proof_main_theorem}
To make the proof easy to read and emphasize important intermediate results, we prove Lemma 
% \ref{lemma:1-norm_to_maxmixed},\ref{lemma:EV_IAOB_VrhoV},\ref{lemma:tracelessOA_OB_sqrtdAdB},\ref{lemma:luo-V1UA-H},\ref{lemma:luo-UAV2-H} 
\ref{lemma:1-norm_to_maxmixed}-\ref{lemma:luo-UAV2-H}
first and derive Theorem~\textcolor{blue}{1} by use of these lemmas.
\begin{lemma}\label{lemma:1-norm_to_maxmixed}
Suppose $V\in\mathbb{V}$ is a unitary on the Hilbert space $\mathcal{H}_A\otimes \mathcal{H}_B$ with $\dim{(\mathcal{H}_A)}=d_A$ and $\dim{(\mathcal{H}_B)}=d_B$ where $\mathbb{V}$ is a unitary $2$-design. Let $\rho$ be an arbitrary density matrix on $\mathcal{H}_A\otimes \mathcal{H}_B$ and $\rho_A=\tr_B(V \rho V^\dagger)$ be the reduced density matrix on $\mathcal{H}_A$ from $V\rho V^\dagger$. The expectation of the $2$-norm distance between $\rho_A$ and the maximally mixed state $I_A/d_A$ satisfies
\begin{equation}\label{Eq:rhoA-Id1<sqrtdA-dB}
    \mathbb{E}_{V}\left\|\rho_{A}-\frac{I_{A}}{d_A}\right\|_{2} \leq \sqrt{\frac{1}{d_B}}.
\end{equation}
\end{lemma}
\begin{proof}
According to the concavity of the square root function and Jensen's inequality in Lemma \ref{lemma:jensen_inequality}, we have
\begin{equation}\label{Eq:rhoA-Id1}
    \mathbb{E}_{V}\left\|\rho_{A}-\frac{I_{A}}{d_A}\right\|_{2}
    \leq \sqrt{\mathbb{E}_{V} \left[ \left\| \rho_{A}-\frac{I_{A}}{d_A} \right\|^2_{2} \right] }.
\end{equation}
Using Corollary \ref{corollary:ave_purity_from_mixed}, the expectation under the square root on the right hand side of (\ref{Eq:rhoA-Id1}) can be exactly calculated as
\begin{equation}\label{Eq:rhoA-Id2=trrho2}
    \begin{aligned}
        \mathbb{E}_{V} \left[ \left\| \rho_{A} - \frac{I_{A}}{d_A} \right\|_{2}^{2} \right]
        & = \mathbb{E}_{V} \tr \left[ \left(\rho_{A} - \frac{I_{A}}{d_A}\right)^{2} \right] 
        = \mathbb{E}_{V} \tr\left(\rho_{A}^{2}-\frac{2}{d_A} \rho_{A}+\frac{I_{A}}{d_A^{2}}\right) \\
        & = \frac{(d_A^2 - 1)d_B}{ d^2 - 1} \tr(\rho^2) + \frac{(d_B^2 - 1)d_A}{d^2 - 1} - \frac{1}{d_A}.
    \end{aligned}
\end{equation}
By the upper bound of the purity $\tr(\rho^2)\leq 1$, (\ref{Eq:rhoA-Id2=trrho2}) could be further relaxed to
\begin{equation}\label{Eq:rhoA-Id2<1dB}
    \begin{aligned}
        \mathbb{E}_{V} \left[ \left\| \rho_{A} - \frac{I_{A}}{d_A} \right\|_{2}^{2} \right]
        & \leq \frac{(d_A^2 - 1)d_B}{ d^2 - 1} + \frac{(d_B^2 - 1)d_A}{d^2 - 1} - \frac{1}{d_A} \\
        & = \frac{d_A + d_B}{d_A d_B + 1} - \frac{1}{d_A} \leq \frac{1}{d_B}.
    \end{aligned}
\end{equation}
Combining (\ref{Eq:rhoA-Id1}) and (\ref{Eq:rhoA-Id2<1dB}), we arrive at 
(\ref{Eq:rhoA-Id1<sqrtdA-dB}).
\end{proof}

\begin{lemma}\label{lemma:EV_IAOB_VrhoV}
Suppose $V \in\mathbb{V}$ is a unitary on the Hilbert space $\mathcal{H}_A\otimes \mathcal{H}_B$ with $\dim{(\mathcal{H}_A)}=d_A$ and $\dim{(\mathcal{H}_B)}=d_B$ where $\mathbb{V}$ is a unitary $2$-design. For any density matrix $\rho$ on $\mathcal{H}_A\otimes \mathcal{H}_B$ and any traceless operator $O_B$ on $\mathcal{H}_B$, the following inequality holds
\begin{equation}\label{Eq:EV_IAOB_VrhoV}
    \mathbb{E}_{V} \left\|\tr_B\left((I_A\otimes O_B)V \rho V^\dagger \right)\right\|_2 
    \leq \| O_B \|_{\infty} \sqrt{ \frac{1}{d_B} }.
\end{equation}
\end{lemma}
\begin{proof}
According to the concavity of the square root function and Jensen's inequality in Lemma \ref{lemma:jensen_inequality}, we have
\begin{equation}\label{Eq:EV_IAOB_VrhoV_1to2}
    \mathbb{E}_{V} \left\|\tr_B\left((I_A\otimes O_B) V\rho V^\dagger \right)\right\|_2
    \leq \sqrt{ \mathbb{E}_{V} \left[ \left\|\tr_B\left((I_A\otimes O_B) V\rho V^\dagger \right)\right\|_2^2 \right] }.
\end{equation}
Using Corollary \ref{corollary:EV_IAtracelessOBVrhoV}, the expectation under the square root in (\ref{Eq:EV_IAOB_VrhoV_1to2}) can be exactly calculated as
\begin{equation}\label{Eq:EV_IAOB_VrhoV_22}
    \begin{aligned}
        \mathbb{E}_{V} \left[ \left\|\tr_B\left((I_A\otimes O_B) V\rho V^\dagger \right)\right\|_2^2 \right]
        = \frac{d_A^2 \|O_B\|^2_2}{d^2-1} \left( \tr(\rho^2) -\frac{1}{d} \right),
    \end{aligned}
\end{equation}
By the upper bound of the purity $\tr(\rho^2)\leq 1$, (\ref{Eq:EV_IAOB_VrhoV_22}) could be further relaxed to
\begin{equation}
    \begin{aligned}
        \mathbb{E}_{V} \left[ \left\|\tr_B\left((I_A\otimes O_B) V\rho V^\dagger \right)\right\|_2^2 \right]
        \leq \frac{d_A^2 \|O_B\|^2_2}{d^2-1} \left( 1 -\frac{1}{d} \right)
        = \frac{d_A^2 \|O_B\|^2_2}{d(d+1)} \leq \frac{\|O_B\|_2^2}{d_B^2}.
    \end{aligned}
\end{equation}
Considering $\|O_B\|_2 \leq \sqrt{d_B} \|O_B\|_{\infty}$, we further obtain
\begin{equation}\label{Eq:EV_IAOB_infty}
  \mathbb{E}_{V} \left[ \left\|\tr_B\left((I_A\otimes O_B) V\rho V^\dagger \right)\right\|_2^2 \right]
  \leq \frac{\|O_B\|_{\infty}^2}{d_B}.
\end{equation}
Combining (\ref{Eq:EV_IAOB_VrhoV_1to2}) and (\ref{Eq:EV_IAOB_infty}), we arrive at (\ref{Eq:EV_IAOB_VrhoV}).
\end{proof}

\begin{lemma}\label{lemma:tracelessOA_OB_sqrtdAdB}
Suppose $V\in\mathbb{V}$ is a unitary on the Hilbert space $\mathcal{H}_A\otimes \mathcal{H}_B$ with $\dim{(\mathcal{H}_A)}=d_A$ and $\dim{(\mathcal{H}_B)}=d_B$ where $\mathbb{V}$ is a unitary $2$-design. Let $O_A$ be an arbitrary traceless operator on $\mathcal{H}_A$ and $O_B$ be either an arbitrary traceless operator or a homothety $c I_B$ on $\mathcal{H}_B$, where $I_B$ is the identity operator on $\mathcal{H}_B$ and $c\in\mathbb{C}$ is an arbitrary complex number. Denote $U_A\in\mathcal{U}(d_A)$ as a unitary operator on $\mathcal{H}_A$. For any density matrix $\rho$ on $\mathcal{H}_A\otimes \mathcal{H}_B$, the following inequality holds
\begin{equation}\label{Eq:tracelessOA_OB_sqrtdAdB}
    \mathbb{E}_{V} \left[ \max_{U_A} \left|
    \tr \left[ (O_A \otimes O_B) ( U_A \otimes I_B ) V \rho V^\dagger ( U_A^\dagger \otimes I_B ) \right] \right|
    \right] 
    \leq \|O_A\|_{\infty} \|O_B\|_{\infty} \sqrt{\frac{ d_A }{ d_B }}.
\end{equation}
\end{lemma}
\begin{proof}
The trace expression on the left hand side of (\ref{Eq:tracelessOA_OB_sqrtdAdB}) can be rewritten as
\begin{equation}\label{Eq:tracelessOA_OB_rewritten}
    \begin{aligned}
        \tr \left[ (O_A \otimes O_B) ( U_A \otimes I_B ) V \rho V^\dagger ( U_A^\dagger \otimes I_B ) \right] 
        = \tr \left[ (U_A^\dagger O_A U_A) \tr_B \left( (I_A \otimes O_B) V \rho V^\dagger \right) \right].
    \end{aligned}
\end{equation}
On the one hand, if $O_B$ is traceless, by using Hölder's inequality in Lemma \ref{lemma:holder}, we obtain
\begin{equation}\label{Eq:tracelessOB_holder}
    \begin{aligned}
        \left| \tr \left[ (U_A^\dagger O_A U_A) \tr_B \left( (I_A \otimes O_B) V \rho V^\dagger \right) \right] \right|
        & \leq \left\| U_A^\dagger O_A U_A \right\|_{2} \left\| \tr_B \left( (I_A \otimes O_B) V \rho V^\dagger \right) \right\|_{2} \\
        & \leq \sqrt{d_A} \left\| O_A \right\|_{\infty} \left\| \tr_B \left( (I_A \otimes O_B) V \rho V^\dagger \right) \right\|_{2},
    \end{aligned}
\end{equation}
where we have used the unitary invariance of the Schatten norms to eliminate $U_A$. Since (\ref{Eq:tracelessOB_holder}) holds for any $U_A$, it certainly holds when taking the maximum, i.e. 
\begin{equation}\label{Eq:tracelessOB_holder_maxUA}
    \begin{aligned}
        \max_{U_A} \left| \tr \left[ (U_A^\dagger O_A U_A) \tr_B \left( (I_A \otimes O_B) V \rho V^\dagger \right) \right] \right|
        \leq \sqrt{d_A} \left\| O_A \right\|_{\infty} \left\| \tr_B \left( (I_A \otimes O_B) V \rho V^\dagger \right) \right\|_{2}.
    \end{aligned}
\end{equation}
Together with Lemma \ref{lemma:EV_IAOB_VrhoV}, we arrive at
\begin{equation}\label{Eq:tracelessOB_sqrtdAdB}
    \begin{aligned}
        & \mathbb{E}_{V} \left[ \max_{U_A} \left|
        \tr \left[ (O_A \otimes O_B) ( U_A \otimes I_B ) V \rho V^\dagger ( U_A^\dagger \otimes I_B ) \right] \right| \right] \\
        & \leq \sqrt{d_A} \left\| O_A \right\|_{\infty} \mathbb{E}_V \left\| \tr_B \left( (I_A \otimes O_B) V \rho V^\dagger \right) \right\|_{2}
        \leq \|O_A\|_{\infty} \|O_B\|_{\infty} \sqrt{\frac{d_A}{d_B}}.
    \end{aligned}
\end{equation}
On the other hand, if $O_B = c I_B$, the right hand side of (\ref{Eq:tracelessOA_OB_rewritten}) can be further rewritten as
\begin{equation}\label{Eq:homothetyOB_rewritten}
    \begin{aligned}
        \tr \left[ (U_A^\dagger O_A U_A) \tr_B \left( (I_A \otimes O_B) V \rho V^\dagger \right) \right] 
        = c\cdot \tr \left[ U_A^\dagger O_A U_A \rho_A \right]
        = c\cdot \tr \left[ U_A^\dagger O_A U_A \left(\rho_A - \frac{I_A}{d_A} \right) \right].
    \end{aligned}
\end{equation}
where we have used the traceless condition of $O_A$ and $\rho_A=\tr_B(V \rho V^\dagger)$ is the reduced density matrix on $\mathcal{H}_A$ from $V \rho V^\dagger$. Again, by Hölder's inequality in Lemma \ref{lemma:holder}, we obtain
\begin{equation}\label{Eq:homothetyOB_holder}
    \begin{aligned}
        \left| c \cdot \tr \left[ U_A^\dagger O_A U_A \left(\rho_A - \frac{I_A}{d_A} \right) \right] \right|
        & \leq |c| \left\| U_A^\dagger O_A U_A \right\|_{2} \left\| \rho_A - \frac{I_A}{d_A} \right\|_2 \\
        & \leq \sqrt{d_A} \left\| O_B \right\|_{\infty} \left\| O_A \right\|_{\infty} \left\| \rho_A - \frac{I_A}{d_A} \right\|_2.
    \end{aligned}
\end{equation}
Since (\ref{Eq:homothetyOB_holder}) holds for any $U_A$, it certainly holds when taking the maximum. Together with (\ref{Eq:tracelessOA_OB_rewritten}), (\ref{Eq:homothetyOB_rewritten}) and Lemma \ref{lemma:1-norm_to_maxmixed}, we arrive at
\begin{equation}\label{Eq:homothetyOB_sqrtdAdB}
    \begin{aligned}
        & \mathbb{E}_{V} \left[ \max_{U_A} \left|
        \tr \left[ (O_A \otimes O_B) ( U_A \otimes I_B ) V \rho V^\dagger ( U_A^\dagger \otimes I_B ) \right] \right| \right] \\
        & \leq \sqrt{d_A} \left\| O_B \right\|_{\infty} \left\| O_A \right\|_{\infty} \mathbb{E}_V \left\| \rho_A - \frac{I_A}{d_A} \right\|_2
        \leq \left\| O_B \right\|_{\infty} \left\| O_A \right\|_{\infty} \sqrt{\frac{d_A}{d_B}}.
    \end{aligned}
\end{equation}
Combining (\ref{Eq:tracelessOB_sqrtdAdB}) and (\ref{Eq:homothetyOB_sqrtdAdB}), we know that (\ref{Eq:tracelessOA_OB_sqrtdAdB}) holds whether $O_B$ is traceless or $O_B=c I_B$, $c\in\mathbb{C}$.
\end{proof}

\begin{lemma}\label{lemma:OAOBVtracelessHV_sqrtdAdB}
Suppose $V\in\mathbb{V}$ is a unitary on the Hilbert space $\mathcal{H}_A\otimes \mathcal{H}_B$ with $\dim{(\mathcal{H}_A)}=d_A$ and $\dim{(\mathcal{H}_B)}=d_B$ where $\mathbb{V}$ is a unitary $2$-design. Let $O_A$,$O_B$ be arbitrary linear operators on $\mathcal{H}_A$,$\mathcal{H}_B$, respectively. Denote $U_A\in\mathcal{U}(d_A)$ as a unitary operator on $\mathcal{H}_A$. For any traceless matrix $H$ on $\mathcal{H}_A\otimes \mathcal{H}_B$, the following inequality holds
\begin{equation}\label{Eq:OAOBVtracelessHV_sqrtdAdB}
    \mathbb{E}_{V} \left\| \tr \left( (O_A \otimes O_B) V H V^\dagger \right) \right\|_2  
    \leq \|O_A\|_{2} \|O_B\|_{2} \|H\|_{\infty} \sqrt{\frac{d_A}{d-1}},
\end{equation}
where $d=d_Ad_B$ denotes the dimension of the whole Hilbert space $\mathcal{H}_A\otimes \mathcal{H}_B$.
\end{lemma}
\begin{proof}
According to the concavity of the square root function and Jensen's inequality in Lemma \ref{lemma:jensen_inequality}, we have
\begin{equation}\label{Eq:OAOBtracelessH_1to2}
    \mathbb{E}_{V} \left\| \tr_B \left( (O_A\otimes O_B) V H V^\dagger \right)\right\|_2
    \leq \sqrt{ \mathbb{E}_{V} \left[ \left\| \tr_B \left( (O_A\otimes O_B) V H V^\dagger \right) \right\|^2_2 \right]}.
\end{equation}
Using Corollary \ref{corollary:EV_2-norm_OAOBtracelessP}, the expectation under the square root in (\ref{Eq:OAOBtracelessH_1to2}) can be exactly calculated as
\begin{equation}\label{Eq:OAOBtracelessH_integral}
    \mathbb{E}_{V} \left[ \left\| \tr_B \left( (O_A\otimes O_B) V H V^\dagger \right) \right\|^2_2 \right]
    = \frac{\|O_A\|_2^2 \|H\|_2^2}{d^2-1} \left[ d_A \|O_B\|_2^2 - \frac{|\tr O_B|^2}{d} \right].
\end{equation}
Combining (\ref{Eq:OAOBtracelessH_1to2}), (\ref{Eq:OAOBtracelessH_integral}) and $\|H\|_2\leq \sqrt{d} \|H\|_{\infty}$, we arrive at
\begin{equation}
    \begin{aligned}
        \mathbb{E}_{V} \left\|\tr_B \left( (O_A\otimes O_B) V H V^\dagger \right) \right\|_2 
        & \leq \sqrt{\frac{1}{d^2 - 1}} \|O_A\|_2 \|H\|_2 \sqrt{ d_A \|O_B\|_2^2 - \frac{|\tr O_B|^2 }{d} } \\
        & \leq \sqrt{\frac{d_A}{d^2-1}} \|O_A\|_2 \|O_B\|_2 \|H\|_2 
        \leq \sqrt{\frac{d_A}{d-1}} \|O_A\|_2 \|O_B\|_2 \|H\|_{\infty},
    \end{aligned}
\end{equation}
which is exactly the same as (\ref{Eq:OAOBVtracelessHV_sqrtdAdB}).
\end{proof}

\begin{lemma}\label{lemma:luo-V1UA-H}
{\rm (Local unitary behind $2$-design circuit)}
Suppose $V\in\mathbb{V}$ is a unitary on the Hilbert space $\mathcal{H}_A\otimes \mathcal{H}_B$ with $\dim{(\mathcal{H}_A)}=d_A$ and $\dim{(\mathcal{H}_B)}=d_B$ where $\mathbb{V}$ is a unitary $2$-design. Denote $U_A\in\mathcal{U}(d_A)$ as a unitary operator on $\mathcal{H}_A$. For any density matrix $\rho$ and any traceless Hermitian operator $H$ on $\mathcal{H}_A \otimes \mathcal{H}_B$, the following inequality holds
\begin{equation}\label{Eq:luo-V1UA-H}
    \mathbb{E}_{V} \left[ \max_{U_A} \left[
    \tr \left( H ( U_A \otimes I_B ) V \rho V^\dagger ( U_A^\dagger \otimes I_B )\right) \right] \right]
    \leq \|H\|_{\infty} ( 2 d^2_A - 1) \sqrt{ \frac{ d_A }{ d_B } },
\end{equation}
\end{lemma}
\begin{proof}
Any traceless Hermitian operator $H$ could be expanded as
\begin{subequations}
    \begin{align}
        & H = H^A + H^B + H^{AB}, \label{Eq:H_decompose}\\
        & H^A := \tr_B(H) \otimes \frac{I_B}{d_B}, \label{Eq:H_decompose_HA}\\
        & H^B := \frac{I_A}{d_A} \otimes \tr_A(H), \label{Eq:H_decompose_HB}\\
        & H^{AB} := H - H^A - H^B, \label{Eq:H_decompose_HAB}
    \end{align}
\end{subequations}
where $H^{A}$, $H^{B}$ only act on $\mathcal{H}_A$, $\mathcal{H}_B$ non-trivially, respectively. $H^{AB}$ acts on $\mathcal{H}_A$ and $\mathcal{H}_B$ both non-trivially. Here a linear operator acting $H^{A}$($H^{B}$) non-trivially means that the operator can not be decomposed to the tensor product form of $I_A\otimes Q_B$($Q_A\otimes I_B$) where $Q_B$($Q_A$) is an arbitrary operator on $\mathcal{H}_B$($\mathcal{H}_A$). Denote $\{\Lambda_j^A\}_{j=0}^{d_A^2-1}$ is the set of clock-and-shift matrices~\cite{Singh2018} on $\mathcal{H}_A$ which is an orthogonal basis in the linear operator space with respect to the Hilbert-Schmidt inner product. $\Lambda_j^A$ are all unitary and hence $\|\Lambda_j^A\|_{\infty}=1$. We assume $\Lambda_0^A=I_A$ without loss of generality. Then $\Lambda_j^A$ are all traceless except $\Lambda_0^A$. Thus, $H^{AB}$ could be further expanded in terms of $\Lambda_j^A$ as
\begin{equation}\label{Eq:HAB_decompose_OB}
    H^{AB} =  \sum_{j=1}^{d_A^2-1} \Lambda_j^A\otimes O_j^B.
\end{equation}
where the explicit expression of $O_j^B$ could be derived from (\ref{Eq:H_decompose_HAB}) as
\begin{equation}
    \begin{aligned}
        O_j^B & = \frac{1}{d_A} \tr_A \left( \left( \Lambda^{A\dagger}_j \otimes I_B \right) H^{AB} \right) \\
        & = \frac{1}{d_A}\tr_A((\Lambda_j^{A\dagger}\otimes I_B)H) - \frac{1}{d_A} \tr_A[(\Lambda_j^{A\dagger} \otimes I_B) H^A] \\
        & = \frac{1}{d_A} \tr_A((\Lambda_j^{A\dagger}\otimes I_B)H) - \frac{1}{d_Ad_B} \tr_A (\Lambda_j^{A\dagger} \tr_B(H)) \otimes I_B.
    \end{aligned}
\end{equation}
By definition, $O_j^B$ are all traceless.  Combining (\ref{Eq:H_decompose}) and (\ref{Eq:HAB_decompose_OB}), we expand $H$ as a summation of bipartite tensor product operators. Next, we will take the maximum for each term in the summation to obtain the desired bound, i.e.
\begin{subequations}
    \begin{align}
        & \mathbb{E}_V \left[ \max_{U_A} \left[
        \tr \left( H ( U_A \otimes I_B ) V \rho V^\dagger ( U_A^\dagger \otimes I_B )\right) \right] \right] \\
        \leq & \mathbb{E}_V \left[ \max_{U_A} \left[
        \tr \left( H^A ( U_A \otimes I_B ) V \rho V^\dagger ( U_A^\dagger \otimes I_B )\right) \right] \right] \label{Eq:max_H_A}\\
        & + \mathbb{E}_V \left[ \max_{U_A} \left[
        \tr \left( H^B ( U_A \otimes I_B ) V \rho V^\dagger ( U_A^\dagger \otimes I_B )\right) \right] \right] \label{Eq:max_H_B}\\
        & + \sum_{j=1}^{d_A^2-1} \mathbb{E}_V \left[ \max_{U_A} \left|
        \tr \left( (\Lambda_j^A \otimes O^B_j) ( U_A \otimes I_B ) V \rho V^\dagger ( U_A^\dagger \otimes I_B )\right) \right| \right]. \label{Eq:sum_max_OB}
    \end{align}
\end{subequations}
For (\ref{Eq:max_H_A}) involving $H^A$ from (\ref{Eq:H_decompose_HA}), Lemma \ref{lemma:tracelessOA_OB_sqrtdAdB} together with $\|\tr_B(H)\|_{\infty}\leq d_B\|H\|_{\infty}$ from Lemma \ref{lemma:partial_trace_norm} gives
\begin{equation}\label{Eq:EV_max_HA_relax}
    \begin{aligned}
        \mathbb{E}_{V} \left[ \max_{U_A} \left[ \tr \left( \left( \tr_B(H) \otimes\frac{I_B}{d_B} \right) (U_A\otimes I_B) V \rho V^\dagger (U_A^\dagger\otimes I_B) \right) \right] \right]
        \leq \frac{ \|\tr_B(H)\|_{\infty} }{d_B} \sqrt{\frac{d_A}{d_B}}
        \leq \|H\|_{\infty} \sqrt{\frac{d_A}{d_B}}.
    \end{aligned}
\end{equation}
For (\ref{Eq:max_H_B}) involving $H^B$ from (\ref{Eq:H_decompose_HB}), Lemma \ref{lemma:integral_VAV=trAd} together with the given condition $\tr(H)=0$ gives
\begin{equation}\label{Eq:EV_HB_eq0}
    \begin{aligned}
        & \mathbb{E}_V \left[ \max_{U_A} \left[ \tr \left( (\frac{I_A}{d_A} \otimes \tr_A H) ( U_A \otimes I_B ) V \rho V^\dagger ( U_A^\dagger \otimes I_B )\right) \right] \right] \\
        & = \mathbb{E}_V \left[\tr \left( (\frac{I_A}{d_A} \otimes \tr_A H) V \rho V^\dagger \right) \right] 
        = \frac{\tr(\rho)}{d} \tr(H) = 0.
    \end{aligned}
\end{equation}
For each term in (\ref{Eq:sum_max_OB}) involving $O^B_j$ from (\ref{Eq:H_decompose_HAB}) and (\ref{Eq:HAB_decompose_OB}), Lemma \ref{lemma:tracelessOA_OB_sqrtdAdB} gives
\begin{equation}\label{Eq:EV_OB_norm)}
    \begin{aligned}
        \mathbb{E}_{V} \left[ \max_{U_A} \left[ \tr \left( (\Lambda_j^A\otimes O_j^B) (U_A\otimes I_B) V \rho V^\dagger (U_A^\dagger \otimes I_B) \right) \right] \right]
        \leq \|\Lambda_j^A\|_{\infty} \|O_j^B\|_{\infty} \sqrt{\frac{d_A}{d_B}} 
        = \|O_j^B\|_{\infty}\sqrt{\frac{d_A}{d_B}},
    \end{aligned}
\end{equation}
Here $\|O_j^B\|_{\infty}$ can be bounded using Lemma \ref{lemma:partial_trace_norm} as
\begin{equation}\label{Eq:OB_norm_bound}
    \begin{aligned}
        \|O_j^B\|_{\infty}
        &= \left\|\frac{1}{d_A} \tr_A((\Lambda_j^{A \dagger}\otimes I_B)H) - \frac{1}{d_Ad_B} \tr_A (\Lambda_j^{A \dagger} \tr_B(H)) \otimes I_B \right\|_{\infty} \\
        & \leq \frac{1}{d_A} \left\| \tr_A( (\Lambda_j^{A \dagger} \otimes I_B) H)\right\|_{\infty} + \frac{1}{d_Ad_B} \left\| \tr_A (\Lambda_j^{A \dagger} \tr_B(H)) \otimes I_B \right\|_{\infty} \\
        & \leq \left\| (\Lambda_j^{A \dagger} \otimes I_B) H \right\|_{\infty} + \frac{1}{d_B} \left\| \Lambda_j^{A \dagger} \tr_B(H) \right\|_{\infty} \\
        & = \left\| H \right\|_{\infty} + \frac{1}{d_B} \left\| \tr_B(H) \right\|_{\infty}
        \leq \left\| H \right\|_{\infty} + \left\| H \right\|_{\infty} = 2 \left\| H \right\|_{\infty},
    \end{aligned}
\end{equation}
where we have used the unitarity of $\Lambda_j^{A}$ and the unitary invariance of the Schatten norms. (\ref{Eq:EV_OB_norm)}) and (\ref{Eq:OB_norm_bound}) are summarized as
\begin{equation}\label{Eq:EV_OB_H_norm}
    \mathbb{E}_{V} \left[ \max_{U_A} \left| \tr \left( (\Lambda_j^A\otimes O_j^B) (U_A\otimes I_B) V \rho V^\dagger (U_A^\dagger \otimes I_B) \right) \right| \right]
    \leq 2 \|H\|_{\infty}\sqrt{\frac{d_A}{d_B}}.
\end{equation}
Finally, combining (\ref{Eq:EV_max_HA_relax}), (\ref{Eq:EV_HB_eq0}) and (\ref{Eq:EV_OB_H_norm}), we obtain
\begin{equation}\label{Eq:EV_max_H_final}
    \begin{aligned}
        & \mathbb{E}_{V} \left[ \max_{U_A} \left[
        \tr \left( H ( U_A \otimes I_B ) V \rho V^\dagger ( U_A^\dagger \otimes I_B )\right) \right] \right] \\
        & \leq \|H\|_{\infty} \sqrt{\frac{d_A}{d_B}} + (d_A^2-1)\cdot 2 \|H\|_{\infty} \sqrt{\frac{d_A}{d_B}}
        = (2d_A^2 - 1) \|H\|_{\infty}\sqrt{\frac{d_A}{d_B}},
    \end{aligned}
\end{equation}
which is exactly the desired inequality (\ref{Eq:luo-V1UA-H}).
\end{proof}

\begin{lemma}\label{lemma:luo-UAV2-H}
{\rm (Local unitary before $2$-design circuit)}
Suppose $V\in\mathbb{V}$ is a unitary on the Hilbert space $\mathcal{H}_A\otimes \mathcal{H}_B$ with $\dim{(\mathcal{H}_A)}=d_A$ and $\dim{(\mathcal{H}_B)}=d_B$ where $\mathbb{V}$ is a unitary $2$-design. Denote $U_A\in\mathcal{U}(d_A)$ as a unitary operator on $\mathcal{H}_A$. For any density matrix $\rho$ and any traceless Hermitian operator $H$ on $\mathcal{H}_A \otimes \mathcal{H}_B$, the following inequality holds
\begin{equation}\label{Eq:luo-UAV2-H}
    \mathbb{E}_{V} \left[ \max_{U_A} \left[ \tr \left( H V ( U_A \otimes I_B ) \rho ( U_A^\dagger \otimes I_B ) V^\dagger \right) \right] \right]
    \leq \|H\|_{\infty} \frac{d_A^2}{ \sqrt{ d_A d_B - 1 } }.
\end{equation}
\end{lemma}
\begin{proof}
Similar with the proof of Lemma \ref{lemma:luo-V1UA-H}, we denote $\{\Lambda_j^A\}_{j=0}^{d_A^2-1}$ is the set of clock-and-shift matrices~\cite{Singh2018}. Any density matrix $\rho$ can be expanded in terms of $\Lambda_j^A$ as
\begin{equation}\label{Eq:psi_decomp}
    \rho = \sum_{j=0}^{d_A^2-1} \Lambda_j^A\otimes O_j^B,
\end{equation}
where $O_j^{B}$ can be explicitly expressed as 
\begin{equation}\label{Eq:OB_from_rho}
    O_j^B = \frac{1}{d_A} \tr_A( (\Lambda_j^{A \dagger} \otimes I_B) \rho).
\end{equation}
Next, we will take the maximum for each term in the summation in (\ref{Eq:psi_decomp}) to obtain the desired bound, i.e.
\begin{subequations}\label{Eq:max_rho_decompose}
    \begin{align}
        & \mathbb{E}_{V} \left[ \max_{U_A} \left[ \tr \left( H V ( U_A \otimes I_B ) \rho ( U_A^\dagger \otimes I_B ) V^\dagger \right) \right] \right] \label{Eq:max_rho_decompose_a} \\
        & \leq \sum_{j=0}^{d_A^2-1} \mathbb{E}_{V} \left[ \max_{U_A} \left| \tr \left( H V ( U_A \otimes I_B ) ( \Lambda_j^{A} \otimes O_j^B) ( U_A^\dagger \otimes I_B ) V^\dagger \right) \right| \right] \label{Eq:sum_max_rho_decompose} \\
        & = \sum_{j=0}^{d_A^2-1} \mathbb{E}_{V} \left[ \max_{U_A} \left| \tr \left( U_A \Lambda^A_j U_A^\dagger \tr_B(V^\dagger H V (I_A \otimes O_j^B)) \right) \right| \right] \label{Eq:sum_max_rho_decompose_rewrite}
    \end{align}
\end{subequations}
For each term in (\ref{Eq:sum_max_rho_decompose_rewrite}), we employ Hölder's inequality in Lemma \ref{lemma:holder} to obtain
\begin{equation}\label{Eq:UAV2_holder}
    \begin{aligned}
        \left| \tr \left( U_A \Lambda^A_j U_A^\dagger \tr_B(V^\dagger H V (I_A \otimes O_j^B)) \right) \right|
        & \leq \| U_A \Lambda^A_j U_A^\dagger \|_{2} \left\| \tr_B(V^\dagger H V (I_A \otimes O_j^B) ) \right\|_2 \\
        & \leq \sqrt{d_A} \left\| \tr_B(V^\dagger H V (I_A \otimes O_j^B) ) \right\|_2.
    \end{aligned}
\end{equation}
Since (\ref{Eq:UAV2_holder}) holds for any $U_A$, it certainly holds when taking the maximum, i.e. 
\begin{equation}
    \max_{U_A} \left| \tr \left( U_A \Lambda^A_j U_A^\dagger \tr_B(V^\dagger H V (I_A \otimes O_j^B)) \right) \right| \leq \sqrt{d_A} \left\| \tr_B(V^\dagger H V (I_A \otimes O_j^B) ) \right\|_2.
\end{equation}
Together with Lemma \ref{lemma:OAOBVtracelessHV_sqrtdAdB}, we obtain
\begin{equation}\label{Eq:EV_max_IA_OB_OB2_Hinfty}
    \begin{aligned}
        \mathbb{E}_{V} \left[ \max_{U_A} \left| \tr \left( U_A \Lambda^A_j U_A^\dagger \tr_B(V^\dagger H V (I_A \otimes O_j^B)) \right) \right| \right] 
        & \leq \sqrt{d_A} \mathbb{E}_{V} \left\| \tr_B(V^\dagger H V (I_A \otimes O_j^B) ) \right\|_1 \\
        & \leq \frac{ d_A }{\sqrt{d_Ad_B - 1}} \|O_j^B\|_2 \|H\|_{\infty},
    \end{aligned}
\end{equation}
where $\|O_B\|_2$ can be bounded using (\ref{Eq:OB_from_rho}) and Lemma \ref{lemma:partial_trace_norm} as
\begin{equation}\label{Eq:rho_OB_2norm_bound}
    \begin{aligned}
        \| O_B \|_2 \leq \| O_B \|_1 = \frac{1}{d_A} \| \tr_A((\Lambda_j^{A \dagger}\otimes I_B) \rho) \|_1 \leq \frac{1}{d_A} \| (\Lambda_j^{A \dagger}\otimes I_B) \rho \|_1 = \frac{1}{d_A} \| \rho \|_1 = \frac{1}{d_A},
    \end{aligned}
\end{equation}
where we have used the unitarity of $\Lambda_j^{A}$ and the unitary invariance of the Schatten norms. Combining (\ref{Eq:max_rho_decompose}), (\ref{Eq:EV_max_IA_OB_OB2_Hinfty}) and (\ref{Eq:rho_OB_2norm_bound}), we arrive at
\begin{equation}
    \begin{aligned}
        \mathbb{E}_{V} \left[ \max_{U_A} \left|
        \tr \left( H V ( U_A \otimes I_B ) \rho ( U_A^\dagger \otimes I_B ) V^\dagger \right) \right| \right]
        \leq d_A^2 \cdot \frac{d_A}{\sqrt{d_Ad_B-1}} \cdot \frac{1}{d_A} \|H\|_{\infty}
        = \|H\|_{\infty} \frac{d_A^2}{\sqrt{d_Ad_B-1}},
    \end{aligned}
\end{equation}
which is exactly the same as (\ref{Eq:luo-UAV2-H}).
\end{proof}

In fact, in the proofs of Lemma \ref{lemma:luo-V1UA-H} and \ref{lemma:luo-UAV2-H} above, the clock-and-shift matrices could be replaced by Pauli strings specially for qubit systems. Finally, we provide a proof for Theorem~\textcolor{blue}{1}, which we recall for convenience. Note that compared to Theorem~\textcolor{blue}{1} in the manuscript, here we prove a more general version where the Hilbert space dimension is no more restricted to qubit systems.
\renewcommand\theproposition{\textcolor{blue}{1}}
\setcounter{proposition}{\arabic{proposition}-1}
\begin{theorem}
Suppose $V_1 \in\mathbb{V}_1, V_2 \in\mathbb{V}_2$ are unitaries on the Hilbert space $\mathcal{H}_A\otimes \mathcal{H}_B$ with $\dim{(\mathcal{H}_A)}=d_A$ and $\dim{(\mathcal{H}_B)}=d_B$. Denote $U_A\in\mathcal{U}(d_A)$ as a unitary on $\mathcal{H}_A$. If either $\mathbb{V}_1$ or $\mathbb{V}_2$, or both are unitary $2$-designs, then for any density matrix $\rho$ and any Hermitian operator $H$ on $\mathcal{H}_A \otimes \mathcal{H}_B$, then the following inequality holds
\begin{equation}\label{Eq:main-theorem_appendix}
    \mathbb{E}_{V_1, V_2} [ \Delta_{H,\rho}(V_1,V_2) ] \leq 4w(H) d_A^2\sqrt{\frac{d_A}{d_B}}.
\end{equation}
where $\mathbb{E}_{V_1,V_2}$ denotes the expectation over $\mathbb{V}_1,\mathbb{V}_2$ independently. $w(H) = \lambda_{\max}(H) - \lambda_{\min}(H)$ denotes the spectral width of $H$, where $\lambda_{\max}(H)$ is the maximum eigenvalue of $H$ and $\lambda_{\min}(H)$ is the minimum. 
\end{theorem}
\begin{proof}
By definition, we have $\mathbf{U}=V_2(U_A\otimes I_B)V_1$ and
\begin{equation}\label{Eq:EV1V2_Delta}
    \Delta_{H,\rho}(V_1,V_2) = \max_{U_A} \left[ \tr \left( H \mathbf{U} \rho \mathbf{U}^\dagger \right) \right] - \min_{U_A} \left[ \tr \left( H \mathbf{U} \rho \mathbf{U}^\dagger \right) \right],
\end{equation}
where the maximum and minimum with respect to $U_A$ are taken over the entire unitary group $\mathcal{U}(d_A)$ of degree $d_A$. Without loss of generality, we assume that $H$ is traceless since (\ref{Eq:main-theorem_appendix}) is invariant if $H$ is added by a homothety $H\rightarrow H + c I$, $c\in\mathbb{R}$. Moreover, considering that the minimization term in (\ref{Eq:EV1V2_Delta}) could be written as
\begin{equation}\label{Eq:min_to_max}
    - \min_{U_A} \left[ \tr \left( H \mathbf{U} \rho \mathbf{U}^\dagger \right) \right] = \max_{U_A} \left[ \tr \left( (-H) \mathbf{U} \rho \mathbf{U}^\dagger \right) \right],
\end{equation}
and $w(H)=w(-H)$, in order to prove (\ref{Eq:main-theorem_appendix}), we only need to prove that
\begin{equation}\label{Eq:EV1V2_maxUA}
    \mathbb{E}_{V_1, V_2} \left[ \max_{U_A} \left[ \tr \left( H \mathbf{U}\rho \mathbf{U}^\dagger \right) \right] \right] \leq 2 w(H) d_A^2 \sqrt{\frac{d_A}{d_B}},
\end{equation}
holds for any traceless Hermitian operator $H$. On the one hand, if $\mathbb{V}_1$ is a unitary 2-design, Lemma \ref{lemma:luo-V1UA-H} gives
\begin{equation}\label{Eq:EV1_bound}
    \begin{aligned}
        \mathbb{E}_{V_1, V_2} \left[ \max_{U_A} \left[ \tr \left( H \mathbf{U}\rho \mathbf{U}^\dagger \right) \right] \right]
        & = \mathbb{E}_{V_2} \left\{ \mathbb{E}_{V_1} \left[ \max_{U_A} \left[\tr \left( V_2^\dagger H V_2( U_A \otimes I_B ) V_1 \rho V_1^\dagger ( U_A^\dagger \otimes I_B )\right) \right] \right] \right\} \\
        & \leq \mathbb{E}_{V_2} \left[ \| V_2^\dagger H V_2 \|_{\infty} (2 d^2_A - 1) \sqrt{ \frac{ d_A }{ d_B } } \right]
        = \| H \|_{\infty} ( 2 d^2_A - 1) \sqrt{ \frac{ d_A }{ d_B } }.
    \end{aligned}
\end{equation}
where we have used the unitary invariance of the Schatten norms and the normalization condition $\mathbb{E}_{V_2}[1]=1$. On the other hand, if $\mathbb{V}_2$ is a unitary 2-design, Lemma \ref{lemma:luo-UAV2-H} gives
\begin{equation}\label{Eq:EV2_bound}
    \begin{aligned}
        \mathbb{E}_{V_1, V_2} \left[ \max_{U_A} \left[ \tr \left( H \mathbf{U}\rho \mathbf{U}^\dagger \right) \right] \right]
        & = \mathbb{E}_{V_1} \left\{ \mathbb{E}_{V_2} \left[ \max_{U_A} \left[ \tr \left( H V_2 ( U_A \otimes I_B ) V_1 \rho V_1^\dagger ( U_A^\dagger \otimes I_B )V_2^\dagger \right) \right] \right] \right\} \\
        & \leq \mathbb{E}_{V_1} \left\{ \|H\|_{\infty} \frac{d_A^2}{ \sqrt{ d_A d_B - 1 } } \right\}
        = \|H\|_{\infty} \frac{d_A^2}{ \sqrt{ d_A d_B - 1 } }.
    \end{aligned}
\end{equation}
where we have used the fact that $V_1\rho V_1^\dagger$ is also a density matrix and the normalization condition $\mathbb{E}_{V_1}[1]=1$. Note that for any traceless Hermitian operator $H$, we have $\lambda_{\max}(H)\geq 0$, $\lambda_{\min}(H)\leq 0$ and
\begin{equation}\label{Eq:H_infty_wH}
    \|H\|_{\infty} = \max\{ \lambda_{\max}(H), -\lambda_{\min}(H) \} \leq \lambda_{\max}(H) - \lambda_{\min}(H) = w(H).
\end{equation}
Combining (\ref{Eq:EV1_bound}), (\ref{Eq:EV2_bound}), (\ref{Eq:H_infty_wH}) and
\begin{equation}\label{Eq:dAdB_factor_relax}
    \begin{aligned}
        &(2 d^2_A - 1) \sqrt{ \frac{ d_A }{ d_B } } < 2 d^2_A \sqrt{ \frac{ d_A }{ d_B } }, \\
        &\frac{d_A^2}{\sqrt{d_Ad_B-1}} < \frac{d_A^2}{\sqrt{(d_A-1)d_B}} < 2d_A^2 \sqrt{\frac{d_A}{d_B}},
    \end{aligned}
\end{equation}
for $d_A\geq 2$, we know that the inequality
\begin{equation}\label{Eq:EV1V2_maxUA_trHUrhoUdag}
    \mathbb{E}_{V_1, V_2} \left[ \max_{U_A} \left[ \tr \left( H \mathbf{U}\rho \mathbf{U}^\dagger \right) \right] \right] 
    \leq 2 w(H) d_A^2 \sqrt{ \frac{d_A}{d_B} },
\end{equation}
holds if either $\mathbb{V}_1$ or $\mathbb{V}_2$ is a unitary $2$-design. Certainly, (\ref{Eq:EV1V2_maxUA_trHUrhoUdag}) also holds if both $\mathbb{V}_1$ and $\mathbb{V}_2$ are $2$-designs. Together with (\ref{Eq:min_to_max}), we arrive at (\ref{Eq:main-theorem_appendix}).
\end{proof}
\renewcommand{\theproposition}{S\arabic{proposition}}

Note that for qubit systems where $d_A=2^m$ and $d_B=2^{n-m}$, the upper bound in (\ref{Eq:main-theorem_appendix}) reduces to that in the manuscript, i.e.
\begin{equation}
    \mathbb{E}_{V_1, V_2} [ \Delta_{H,\rho}(V_1,V_2) ] \leq \frac{w(H)}{2^{n/2-3m-2}}.
\end{equation}
Although Theorem~\textcolor{blue}{1} only establish an upper bound on the expectation of $\Delta_{H,\rho}(V_1,V_2)$, we can derive the upper bound on the variance of $\Delta_{H,\rho}(V_1,V_2)$ from Theorem~\textcolor{blue}{1} with the non-negativity and boundedness of $\Delta_{H,\rho}(V_1,V_2)$. Namely, since $\Delta_{H,\rho}(V_1,V_2)\in[0,w(H)]$, Lemma \ref{lemma:var<expec} gives
\begin{equation}
    \operatorname{Var}_{V_1,V_2}[\Delta_{H,\rho}(V_1,V_2)] 
    \leq w(H) \cdot \mathbb{E}_{V_1, V_2}[\Delta_{H,\rho}(V_1,V_2)]
    \leq 4 w^2(H) d_A^2 \sqrt{\frac{d_A}{d_B}}.
\end{equation}
Furthermore, Theorem~\textcolor{blue}{1} together with the non-negativity of $\Delta_{H,\rho}(V_1,V_2)$ can also provide an upper bound of the probability that $\Delta_{H,\rho}(V_1,V_2)$ deviates from zero. Specifically, according to Theorem~\textcolor{blue}{1} and Markov's inequality in Lemma \ref{lemma:markov_inequality}, the following concentration inequality
\begin{equation}
    \operatorname{Pr}\left[ \Delta_{H,\rho}(V_1,V_2) \geq \epsilon \right] 
    \leq \frac{\mathbb{E}_{V_1, V_2}[\Delta_{H,\rho}(V_1,V_2)]}{\epsilon}
    \leq \frac{4 w(H) d_A^2}{\epsilon} \sqrt{\frac{d_A}{d_B}},
\end{equation}
holds for any $\epsilon>0$. It is worth noticing that the upper bound in (\ref{Eq:main-theorem_appendix}) only involves $w(H)$ and does not depend on any detail of the Hermitian operator $H$. In order to derive this compact and general upper bound in (\ref{Eq:main-theorem_appendix}), we perform many relaxations such as in (\ref{Eq:HAB_decompose_OB}), (\ref{Eq:H_infty_wH}) and (\ref{Eq:dAdB_factor_relax}). Otherwise, if some specific structures about $H$ are known, a more complicated but tighter bound could be obtained as
\begin{equation}\label{Eq:main-theorem_tighter}
    \mathbb{E}_{V_1, V_2} [ \Delta_{H,\rho}(V_1,V_2) ] \leq \max\{ N_A + 2 N_{AB}, d_A \sqrt{\frac{d}{d-1}} \} \cdot \left\| H - \tr(H)\frac{I}{d} \right\|_\infty  \sqrt{\frac{d_A}{d_B}},
\end{equation}
where $N_A\leq 1$ denotes the number of non-vanishing terms in (\ref{Eq:H_decompose_HA}) and $N_{AB}\leq (d_A^2 - 1)$ denotes the number of non-vanishing terms in (\ref{Eq:HAB_decompose_OB}), which can be seen as a ``coupling rank'' or say ``coupling complexity'' between subsystem $A$ and $B$ of the Hamiltonian $H$. The variational quantum eigensolver (VQE) example of the Heisenberg model $\hat{H}$ in the main text has $N_A=0$, $N_{AB}=3$ and that of quantum autoencoder (QAE) has $N_A=1$, $N_{AB}=0$. Therefore, we have two tighter bound for these two examples as
\begin{equation}
    \begin{aligned}
        \text{Heisenberg:}~~~ &\mathbb{E}_{V_1, V_2} [ \Delta_{{\rm VQE}}(V_1,V_2) ] \leq 24\cdot w(\hat{H}) \cdot \frac{1}{2^{n/2}}, \\
        \text{Autoencoder:}~~~ &\mathbb{E}_{V_1, V_2} [ \Delta_{{\rm QAE}}(V_1,V_2) ] \leq \frac{8}{\sqrt{3}} \cdot \frac{1}{2^{n/2}},
    \end{aligned}
\end{equation}
which are used in the figure of the numerical simulation section in the main text.

\subsection{Proof of Proposition \textcolor{blue}{2}}\label{supplementary:proof_prop_qsl}

In this section, we prove Lemma \ref{lemma:upper_bound_Bures}-\ref{lemma:EV1V2_max_F} first and derive Proposition~\textcolor{blue}{2} by use of these lemmas.
\begin{lemma}\label{lemma:upper_bound_Bures}
For any density matrices $\rho$ and $\sigma$ we have
\begin{equation}\label{Eq:F leq rank HS}
    F(\rho,\sigma)\leq\operatorname{rank}(\rho\sigma)\tr(\rho\sigma),
\end{equation}
where $F(\rho,\sigma)=\left(\tr\sqrt{\rho^{1/2}\sigma\rho^{1/2}}\right)^2$ denotes the Bures fidelity.
\end{lemma}
\begin{proof}
Let $\lambda_i$ be the $i$-th eigenvalue of $\sqrt{\rho^{1/2}\sigma\rho^{1/2}}$ in the non-increasing order. Note that $\lambda_i\geq 0$ holds for any $i$ due to the positive semi-definite property of $\sqrt{\rho^{1/2}\sigma\rho^{1/2}}$. By definition, the square root of the Bures fidelity can be represented as
\begin{equation}\label{Eq:bures as lambda}
    \sqrt{F(\rho,\sigma)} = \sum_i\lambda_i,
\end{equation}
while the square root of the Hilbert-Schmidt inner product of $\rho$ and $\sigma$ can be represented as
\begin{equation}\label{Eq:hs as lambda}
    \sqrt{\tr(\rho\sigma)} = \sqrt{\tr(\rho^{1/2}\sigma\rho^{1/2})} = \sqrt{\sum_i\lambda_i^2}.
\end{equation}
According to the inequality between the vector $1$-norm and $2$-norm $\|\mathbf{x}\|_1\leq\sqrt{n}\|\mathbf{x}\|_2$ for any $n$-dimensional vector $\mathbf{x}$, (\ref{Eq:bures as lambda}) and (\ref{Eq:hs as lambda}) lead to
\begin{equation}
    \sqrt{F(\rho,\sigma)}\leq\sqrt{\operatorname{rank}(\rho\sigma)}\sqrt{\tr(\rho\sigma)}.
\end{equation}
Take the square of both sides and we arrive at (\ref{Eq:F leq rank HS}).
\end{proof}

\begin{lemma}\label{lemma:EVmaxUA_F_dAdB}
Suppose $V\in\mathbb{V}$ is a unitary on the Hilbert space $\mathcal{H}_A\otimes \mathcal{H}_B$ with $\dim{(\mathcal{H}_A)}=d_A$ and $\dim{(\mathcal{H}_B)}=d_B$ where $\mathbb{V}$ is a unitary $1$-design. Denote $U_A\in\mathcal{U}(d_A)$ as a unitary operator on $\mathcal{H}_A$. For any density matrices $\rho$ and $\sigma$ on $\mathcal{H}_A\otimes \mathcal{H}_B$, the following inequality holds
\begin{equation}\label{Eq:EVmaxUA_F_dAdB}
    \mathbb{E}_{V} \left[ 
    \max_{U_A} F ( (U_A \otimes I_B) V \rho V^\dagger (U_A \otimes I_B)^\dagger, \sigma ) 
    \right] \leq \frac{d_A}{d_B}.
\end{equation}
where $F$ denotes the Bures fidelity.
\end{lemma}
\begin{proof}
According to the monotonicity of the Bures fidelity under the action of quantum channels~\cite{Watrous2018}, we have
\begin{equation}\label{Eq:fidelity_monotonicity}
    F ( (U_A \otimes I_B) V \rho V^\dagger (U_A \otimes I_B)^\dagger, \sigma ) 
    \leq F ( \tr_A \left( (U_A \otimes I_B) V \rho V^\dagger (U_A \otimes I_B)^\dagger \right), \tr_A\sigma ) 
    = F ( \tr_A \left( V \rho V^\dagger \right), \tr_A\sigma ).
\end{equation}
Since (\ref{Eq:fidelity_monotonicity}) holds for any $U_A$, it certainly holds when taking the maximum. Together with Lemma \ref{lemma:upper_bound_Bures}, it holds that
\begin{equation}\label{Eq:fidelity_monotonicity_max}
    \max_{U_A} F ( (U_A \otimes I_B) V \rho V^\dagger (U_A \otimes I_B)^\dagger, \sigma )
    \leq F ( \tr_A \left( V \rho V^\dagger \right), \tr_A\sigma )
    \leq d_A \tr( \tr_A \left( V \rho V^\dagger \right) \tr_A\sigma ).
\end{equation}
Because $\mathbb{V}$ is a unitary $1$-design, we can apply Lemma \ref{lemma:integral_VAV=trAd} to obtain
\begin{equation}\label{Eq:EV_trtrAVrhoVtrAsigma}
    \begin{aligned}
        \mathbb{E}_{V} \left[ \tr( \tr_A \left( V \rho V^\dagger \right) \tr_A\sigma ) \right] 
        = \tr\left( \tr_A \left( \frac{\tr(\rho)}{d} I \right) \tr_A\sigma \right) 
        = \tr(\rho) \tr(\sigma) \frac{1}{d_B} \leq \frac{1}{d_B},
    \end{aligned}
\end{equation}
where $d=d_Ad_B$ denotes the dimension of $\mathcal{H}_A\otimes \mathcal{H}_B$. Combining (\ref{Eq:fidelity_monotonicity_max}) and (\ref{Eq:EV_trtrAVrhoVtrAsigma}), we arrive at (\ref{Eq:EVmaxUA_F_dAdB}).
\end{proof}

\begin{lemma}\label{lemma:EV1V2_max_F}
Suppose $V_1 \in\mathbb{V}_1, V_2 \in\mathbb{V}_2$ are independent unitaries on the Hilbert space $\mathcal{H}_A\otimes \mathcal{H}_B$ with $\dim{(\mathcal{H}_A)}=d_A$ and $\dim{(\mathcal{H}_B)}=d_B$. Denote $U_A\in \mathcal{U}(d_A)$ as a unitary operator on $\mathcal{H}_A$. If either $\mathbb{V}_1$ or $\mathbb{V}_2$, or both are unitary $1$-designs, then for any density matrix $\rho$ and $\sigma$ on $\mathcal{H}_A \otimes \mathcal{H}_B$, the following inequality holds
\begin{equation}\label{Eq:EV1V2_max_F}
    \mathbb{E}_{V_1,V_2}\left[ \max_{U_A} F \left( \mathbf{U} \rho \mathbf{U}^\dagger, \sigma \right) \right] \leq \frac{d_A}{d_B},
\end{equation}
where $\mathbf{U} = V_2 (U_A\otimes I_B) V_1$ and $F$ is the Bures fidelity.
\end{lemma}
\begin{proof}
On the one hand, if $\mathbb{V}_1$ is a unitary $1$-design, Lemma \ref{lemma:EVmaxUA_F_dAdB} gives
\begin{equation}\label{Eq:EV1_qsl}
    \begin{aligned}
        \mathbb{E}_{V_1,V_2}\left[ \max_{U_A} F \left( \mathbf{U} \rho \mathbf{U}^\dagger, \sigma \right) \right]
        & = \mathbb{E}_{V_2} \left\{ \mathbb{E}_{V_1} \left[\max_{U_A} F \left( (U_A\otimes I_B) V_1 \rho V_1^\dagger (U_A \otimes I_B)^\dagger, V_2^\dagger \sigma V_2 \right) \right] \right\} \\
        & \leq \mathbb{E}_{V_2}\left[ \frac{d_A}{d_B} \right] = \frac{d_A}{d_B},
    \end{aligned}
\end{equation}
where we have used the unitary invariance of the fidelity and the normalization condition $\mathbb{E}_{V_2}[1]=1$. Note that in this case there is no restriction on $\mathbb{V}_2$. On the other hand, if $\mathbb{V}_2$ is a unitary $1$-design, similarly, Lemma \ref{lemma:EVmaxUA_F_dAdB} gives
\begin{equation}\label{Eq:EV2_qsl}
    \begin{aligned}
        \mathbb{E}_{V_1,V_2} \left[ \max_{U_A} F \left( \mathbf{U} \rho \mathbf{U}^\dagger, \sigma \right) \right]
        & = \mathbb{E}_{V_1} \left\{ \mathbb{E}_{V_2} \left[\max_{U_A} F \left( V_1 \rho V_1^\dagger, (U_A\otimes I_B)^\dagger V_2^\dagger \sigma V_2 (U_A \otimes I_B) \right) \right] \right\} \\
        & \leq \mathbb{E}_{V_1}\left[ \frac{d_A}{d_B} \right] = \frac{d_A}{d_B},
    \end{aligned}
\end{equation}
where we have used the unitary invariance of the fidelity again and the normalization condition $\mathbb{E}_{V_1}[1]=1$. Combining (\ref{Eq:EV1_qsl}) and (\ref{Eq:EV2_qsl}), we know that (\ref{Eq:EV1V2_max_F}) holds if either $\mathbb{V}_1$ or $\mathbb{V}_2$ is a unitary $1$-design. Certainly, (\ref{Eq:EV1V2_max_F}) also holds if both $\mathbb{V}_1$ and $\mathbb{V}_2$ are $1$-designs.
\end{proof}

Finally, we provide a proof for Proposition \textcolor{blue}{2}. Compared to Proposition \textcolor{blue}{2} in the manuscript, here we prove a more general version where the Hilbert space dimension is no more restricted to qubit systems.
\renewcommand\theproposition{\textcolor{blue}{2}}
\setcounter{proposition}{\arabic{proposition}-1}
\begin{proposition}
Suppose $V_1 \in\mathbb{V}_1, V_2 \in\mathbb{V}_2$ are independent unitaries on the Hilbert space $\mathcal{H}_A\otimes \mathcal{H}_B$ with $\dim{(\mathcal{H}_A)}=d_A$ and $\dim{(\mathcal{H}_B)}=d_B$. Denote $U_A\in \mathcal{U}(d_A)$ as a unitary operator on $\mathcal{H}_A$. If either $V_1$ or $V_2$, or both are from unitary $1$-designs, then for any density matrices $\rho$ and $\sigma$, the following inequality holds
\begin{equation}\label{Eq:state-learning}
    \mathbb{E}_{V_1,V_2} \left[ \Delta_{{\rm Q S L}}(V_1,V_2) \right] \leq \frac{d_A}{d_B},
\end{equation}
where $\mathbb{E}_{V_1,V_2}$ denotes the expectation over $\mathbb{V}_1,\mathbb{V}_2$ independently.
\end{proposition}
\begin{proof}
By definition, we have $\mathbf{U}=V_2 (U_A\otimes I_B) V_1$ and 
\begin{equation}
    \Delta_{{\rm Q S L}}(V_1,V_2) = \max_{U_A} F \left( \mathbf{U} \rho \mathbf{U}^\dagger, \sigma \right) - \min_{U_A} F \left( \mathbf{U} \rho \mathbf{U}^\dagger, \sigma \right).
\end{equation}
According to Lemma \ref{lemma:EV1V2_max_F} and the non-negativity of the fidelity, it holds that
\begin{equation}
    \begin{aligned}
        \mathbb{E}_{V_1,V_2} \left[ \Delta_{{\rm Q S L}}(V_1,V_2) \right]
        \leq \mathbb{E}_{V_1,V_2}\left[ \max_{U_A} F \left( \mathbf{U} \rho \mathbf{U}^\dagger, \sigma \right) \right] \leq \frac{d_A}{d_B},
    \end{aligned}
\end{equation}
if either $V_1$ or $V_2$, or both are from unitary $1$-designs.
\end{proof}

For qubit systems where $d_A=2^m$ and $d_B=2^{n-m}$, the upper bound in (\ref{Eq:state-learning}) reduces to that in the manuscript, i.e.
\begin{equation}
    \mathbb{E}_{V_1, V_2} [ \Delta_{{\rm QSL}} (V_1, V_2) ] \leq \frac{1}{2^{n-2m}}.
\end{equation}
Importantly, due to the non-negativity and boundedness of $\Delta_{{\rm QSL}}(V_1,V_2)$, we can derive the upper bound on the variance and the probability tail from Proposition \textcolor{blue}{2} using Lemma \ref{lemma:var<expec} and Markov's inequality in Lemma \ref{lemma:markov_inequality}, i.e.
\begin{equation}
    \begin{aligned}
        & \operatorname{Var}_{V_1,V_2}[\Delta_{{\rm QSL}}(V_1,V_2)] 
        \leq 1 \cdot \mathbb{E}_{V_1, V_2}[\Delta_{{\rm QSL}}(V_1,V_2)]
        \leq \frac{d_A}{d_B}, \\
        & \operatorname{Pr}\left[ \Delta_{{\rm QSL}}(V_1,V_2) \geq \epsilon \right]
        \leq \frac{\mathbb{E}_{V_1, V_2}[\Delta_{{\rm QSL}}(V_1,V_2)]}{\epsilon}
        \leq \frac{1}{\epsilon} \frac{d_A}{d_B},~~\forall~\epsilon>0.
    \end{aligned}
\end{equation}

%%%%%%%%%%%%%%%%%%%%%%%%%%%%%%%%%%%%%%%%%%%%%%%%%%%%%%%%%%%%%%%%%%%%%%%%%%%%%%
\subsection{Numerical simulation with varying layers}\label{supplementary:varying_layers}
This section provides some experimental results on how the variation range of the cost function caused by a local unitary varies with the number of circuit layers. We construct circuits of $V_1$ with different numbers of layers to perform experiments with other settings the same as those in the manuscript. As shown in Fig.~\ref{fig:with_layers}, different lines with markers represent the average value of $\Delta_{{\rm VQE}}(V_1,V_2)$ over samples vs. the number of qubits $n$ corresponding to different numbers of layers we laid in $V_1$. We can see that as the number of layers increases, these lines become more and more parallel to the dashed reference line, which has a slope of $-0.5$, i.e., the exponential decay rate we derived in Theorem~\textcolor{blue}{1}. Thus there is a transition to $2$-design where $\mathbb{E}_{V_1,V_2}[\Delta_{H,\rho}(V_1,V_2)]$ converges. This implies that Theorem~\textcolor{blue}{1} is valid when the circuit is sufficiently deep, practically with depth around $10\times n$.
\begin{figure}
    \includegraphics[scale=0.4]{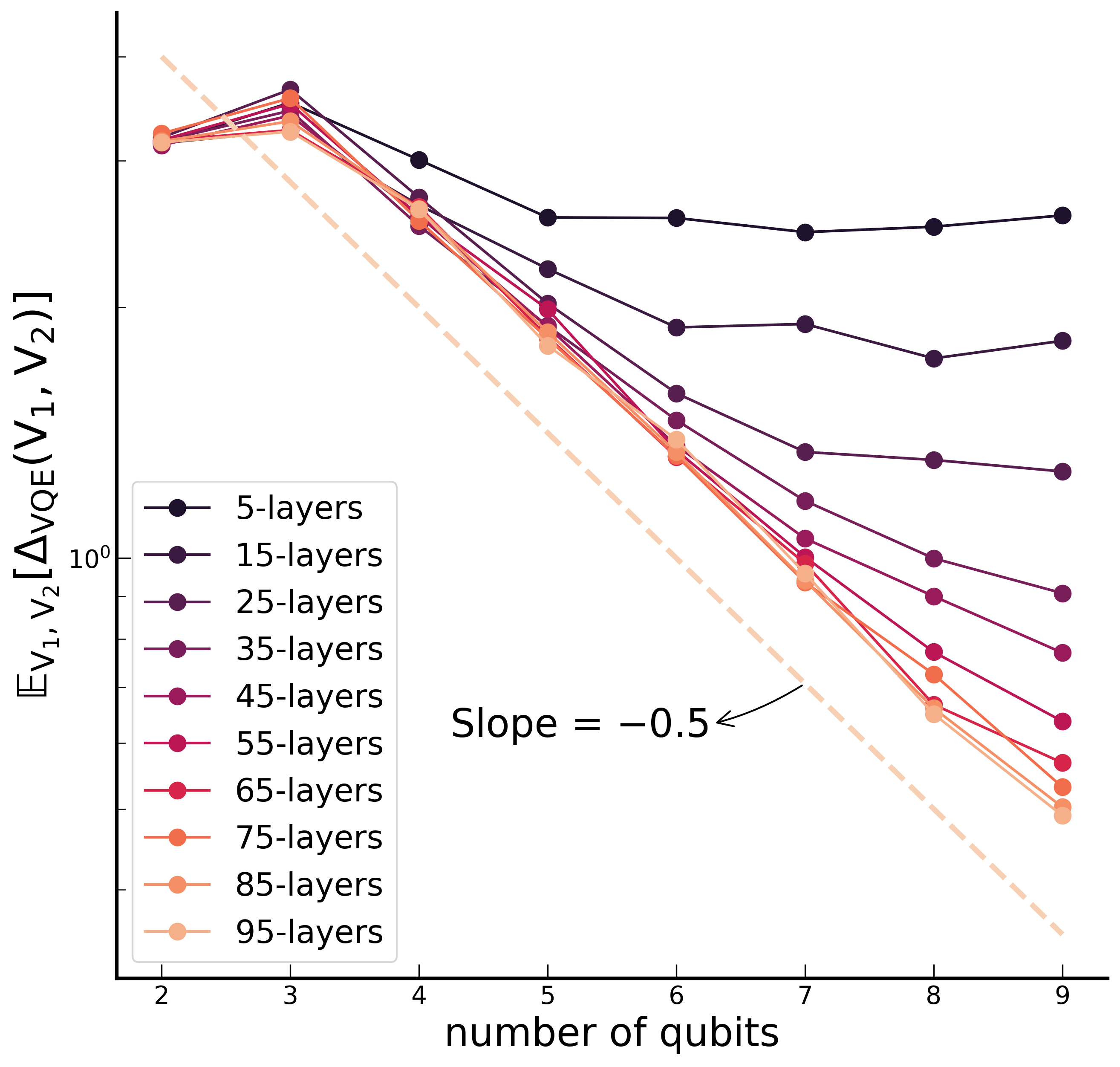}
    \caption{The semi-log plot of the average value of the variation range $\Delta_{{\rm VQE}}(V_1,V_2)$ vs. the number of qubits. The cost function used here is the energy expectation of the $1$-dimensional antiferromagnetic Heisenberg model. Different lines represent different numbers of circuit layers from $5$ to $95$ with step length $10$, with the line for $5$ layers on the top and $95$ layers on the bottom. And the dashed line, as a guide to the eye, has a slope of $-0.5$, which is the exponential decay rate we derived in Theorem~\textcolor{blue}{1}.}
    \label{fig:with_layers}
\end{figure}

\end{document}